\documentclass[reqno]{amsart}
\usepackage{amsmath,amsfonts,amsgen,amstext,amsbsy,amsopn,amsthm, amssymb}

\newtheorem{Lemma}{Lemma} 
\newtheorem{Theorem}{THEOREM}
\newtheorem{prop}{Proposition} 
\newtheorem{assu}{Assumption}
 
\theoremstyle{definition}

\newcommand\calC{{\mathcal{C}}}
\newcommand\Tr{{\rm Tr\,}} 
\newcommand\id{\mathbb{I}}
\newcommand\nn\nonumber

\numberwithin{equation}{section}

\newcommand{\R}{\mathbb{R}}
 
\newcommand{\N}{\mathbb{N}}
\newcommand{\Z}{\mathbb{Z}} 
\newcommand{\C}{\mathbb{C}}
 
\newcommand{\F}{\mathcal{F}}
\newcommand{\E}{\mathcal{E}} 
\renewcommand{\H}{\mathcal{H}}

\DeclareMathOperator{\const}{const}
\DeclareMathOperator{\im}{Im} 
\DeclareMathOperator{\re}{Re}

\begin{document}

\title{Microscopic Derivation of Ginzburg-Landau Theory}

\author[R.L. Frank]{Rupert L. Frank} \address{{\rm (Rupert L. Frank)}
  Department of Mathematics, Princeton University, Princeton, NJ
  08544, USA} \email{rlfrank@math.princeton.edu}

\author[C. Hainzl]{Christian Hainzl} \address{{\rm (Christian Hainzl)}
  Mathematisches Institut, Universit\"at T\"ubingen, Auf der
  Morgenstelle 10, 72076 T\"ubingen, Germany}
\email{christian.hainzl@uni-tuebingen.de}

\author[R. Seiringer]{Robert Seiringer} \address{{\rm (Robert
    Seiringer)} Department of Mathematics and Statistics, McGill
  University, 805 Sherbrooke Street West, Montreal, QC H3A 2K6,
  Canada} \email{rseiring@math.mcgill.ca}

\author[J.P. Solovej]{Jan Philip Solovej} \address{{\rm (Jan Philip
    Solovej)} Department of Mathematics, University of Copenhagen,
  Universitetsparken 5, DK-2100 Copenhagen, Denmark}
\email{solovej@math.ku.dk}

\date{February 19, 2011}

\begin{abstract}
  We give the first rigorous derivation of the celebrated
  Ginzburg-Landau (GL) theory, starting from the microscopic
  Bardeen-Cooper-Schrieffer (BCS) model. Close to the critical
  temperature, GL arises as an effective theory on the macroscopic
  scale. The relevant scaling limit is semiclassical in nature, and
  semiclassical analysis, with minimal regularity assumptions, plays
  an important part in our proof.
\end{abstract}

\maketitle

\section{Introduction and Main Results}

\subsection{Introduction}

In 1950 Ginzburg and Landau \cite{GL} introduced a model of
superconductivity that has been extremely successful and is widely
used in physics, even beyond the theory of superconductivity. It has a
rich mathematical structure that has been studied in great detail, and
has inspired the development of many interesting new
concepts. Ginzburg and Landau arrived at their model in a
phenomenological way, describing the {\em macroscopic} properties of a
superconductor, without the need to understand the microscopic
mechanism.
 
In 1957 Bardeen, Cooper and Schrieffer \cite{BCS} formulated the first
{\em microscopic} explanation of superconductivity starting from a
many-body Hamiltonian. In a major breakthrough they realized that this
phenomenon can be described by a {\em pairing mechanism}. Below a
critical temperature, a superconducting paired state forms due to an
instability of the normal state in the presence of an attraction
between the particles. In the case of a metal, an effective attraction
between the electrons arises through an interaction with phonons, the
quantized vibrations of the underlying lattice formed by atoms.

A connection between the two approaches, the phenomenological GL
theory and the microscopic BCS theory, was established in 1959 when
Gorkov \cite{gorkov} explained how, close to the critical temperature,
the GL theory arises from the BCS model. A simplified version of his
explanation was later given by de Gennes in his textbook
\cite{deGenne}. The purpose of this paper is to give the first fully
rigorous mathematical derivation of GL theory from the BCS model. In
our approach it is not necessary to restrict attention to the
particular interaction considered in the original BCS model, indeed we
allow for a large class of attractions among the particles. In
particular, we allow for local two-body interactions, in contrast to
the simpler (non-local rank one) interaction considered by BCS. Such
local interactions are relevant for the description of cold gases of
fermionic atoms, which are of great current interest. These atoms are
electrically neutral and the corresponding pairing mechanism is
relevant for superfluidity rather than superconductivity.

The BCS model is a considerable simplification of the full many-body
problem. In the latter a state is described by a complicated wave
function of a macroscopic number of variables. In the BCS model, all
the information about the system is contained in quantities depending
on only two variables; the reduced one-particle density matrix
$\gamma$, i.e., a positive trace class operator on the one-particle
space, and the Cooper pair (two particle) wave function $\alpha$,
which is non-zero only below the critical temperature.  The GL model
is yet much simpler, as it describes the system by a single function
$\psi$ of one variable only, which satisfies a non-linear second order
PDE, the GL equation.  This function $\psi$ only describes macroscopic
variations in the system, whereas the BCS states $\gamma$ and $\alpha$
exhibit both microscopic and macroscopic details. Hence GL theory
represents a significant simplification as to the more complicated BCS
theory. The relation of these three theories is analogous to atomic
physics where quantum mechanics, Hartree-Fock theory and Thomas-Fermi
theory are models of the same type of increasing simplicity. In
contrast to the atomic case, where the Hartree-Fock approximation is
mathematically well understood, it remains an open problem to
rigorously establish the relation of the BCS model to the full
many-body quantum-mechanical description.

The BCS theory may be thought of as a variational theory, which
utilizes a special class of trial states known as quasi-free
states. This point of view was emphasized by Leggett \cite{Leg}. The
mathematical aspects of this theory in the translation invariant case
were studied in \cite{FHNS,HHSS,HS,HS3,HS2}.  In this paper we study
the non-translation invariant case, where weak external fields are
present that vary only on the macroscopic scale. We show that close to
the critical temperature $T_c$, where $\alpha$ is small, the
macroscopic variations of $\alpha$ are correctly described by the GL
theory.  More precisely, we assume that the microscopic scale of the
system is of order $h\ll 1$ relative to the macroscopic scale.
Variations of the system on the macroscopic scale cause a relative
change of the energy of order $h^2$, hence the external fields should
change the energy to the same order. This change is correctly
described by the GL theory to leading order in $h$ if $(T_c-T)/T_c$ is
of the order $h^2$, with $T$ the actual temperature. We also prove
that the GL wave function $\psi$ correctly describes the macroscopic
behavior of the BCS state. The BCS Cooper pair wave function in this
parameter regime equals
\begin{equation}\label{alphconverg}
  \alpha(x,y) = \tfrac 12\left( \psi(hx) + \psi(hy)\right)  \alpha_0(x-y) 
\end{equation}
to leading order in $h$, where $\alpha_0$ is the unperturbed
translation invariant pair function. Equivalently, $\frac
12(\psi(hx)+\psi(hy))$ could be replaced by $\psi(h (x+y)/2)$ to
leading order. In particular, $\psi$ describes the center-of-mass
motion of the BCS state. The role of $h$ as a semiclassical parameter
can be understood as follows. If one rescales the arguments to
macroscopic variables and thinks of $\alpha(\bar x/h,\bar y/h)$ as the
integral kernel of an operator, it is the quantization of the
semiclassical symbol $\psi(\bar x) \widehat \alpha_0(p)$. It would, in
fact, be the Weyl quantization if we had used the alternative
center-of-mass representation. For technical reasons we found it more
convenient to work with the representation (\ref{alphconverg}).

Our approach is motivated by the work of de Gennes
\cite{deGenne}. While de Gennes studied the emergence of the GL
equation from the BCS equations, it is important for our rigorous
analysis to utilize that these equations are Euler-Lagrange equations
for variational problems. We give precise bounds on the lowest energy
of the BCS functional and the connection of the corresponding
minimizer with the GL minimizer. The exact statement of our result is
given in Subsection~\ref{ss:results}.

\subsection{The BCS Functional}

We consider a macroscopic sample of a fermionic system, in $d$ spatial
dimensions, where $1\leq d\leq 3$. Let $\mu\in\R$ denote the chemical
potential and $T>0$ the temperature of the sample. The fermions
interact through a local two-body potential $V$. In addition, they are
subject to external electric and/or magnetic fields. Neutral atoms
would not couple to these fields, of course, but there can be other
forces, e.g., arising from rotation, with a similar mathematical
description.  In BCS theory the state of the system can be
conveniently described in terms of a $2\times 2$ operator valued
matrix
\begin{equation}\label{def:gamma}
  \Gamma = \left( \begin{array}{cc} \gamma & \alpha \\ \bar\alpha & 1 -\bar\gamma \end{array}\right) 
\end{equation}
satisfying $0\leq \Gamma \leq 1$ as an operator on $L^2(\R^d)\oplus
L^2(\R^d)\cong L^2(\R^d)\otimes \C^2$. The bar denotes complex
conjugation, i.e., $\bar \alpha$ has the integral kernel $\overline
{\alpha(x,y)}$. In particular, $\Gamma$ is hermitian, implying that
$\gamma$ is hermitian and $\alpha$ is symmetric, i.e., $\gamma(x,y) =
\overline{\gamma(y,x)}$ and $\alpha(x,y)=\alpha(y,x)$. There are no
spin variables in $\Gamma$. The full, spin dependent Cooper pair wave
function is the product of $\alpha$ with an antisymmetric spin
singlet. This is why $\alpha$ itself is symmetric so that we obtain
the antisymmetric fermionic character of the full, spin-dependent,
pair wave function.

The general form of the BCS functional for the free energy of such a
system is, in suitable units,
$$
\Tr \left[ \left( \left(-i \nabla + A(x) \right)^2 -\mu + W(x)\right)
  \gamma \right] - T\, S(\Gamma) + \int V(x-y) |\alpha(x,y)|^2 \, {dx
  \, dy} \,.
$$ 
Here, $A$ is the magnetic vector potential, and $W$ is the external
electric potential.  The entropy equals
\begin{equation}
  S(\Gamma) = - \Tr \left[ \Gamma \ln \Gamma \right]\,,
\end{equation}
where the trace is now both over $\C^2$ and $L^2(\R^3)$.  The BCS
state of the system is a minimizer of this functional over all
admissible states $\Gamma$.

As explained in detail in \cite[Appendix~A]{HHSS}, the BCS functional
can be heuristically derived from the full many-body Hamiltonian via
two steps of simplification. First, one considers only quasi-free
states, and second one neglects the resulting direct and exchange term
in the interaction energy. The latter terms are considered unimportant
in the physically relevant parameter regimes \cite{bookleggett}.

We are interested in the effect of weak and slowly varying external
fields, as already explained in the previous subsection. More
precisely, $A(x)$ should be replaced by $h A(hx)$ and $W(x)$ by $h^2
W(hx)$.  In order to avoid having to introduce boundary conditions, we
assume that the system is infinite and periodic with period $h^{-1}$,
in all $d$ directions. In particular, $A$ and $W$ should be
periodic. We also assume that the state $\Gamma$ is periodic. The aim
then is to calculate the free energy per unit volume.

We find it convenient to do a rescaling and use macroscopic variables
instead of the microscopic ones. The rescaled BCS functional has the
form
\begin{align}\nonumber
  \F^{\rm BCS}(\Gamma) & := \Tr \left[ \left( \left(-i h \nabla + h
        A(x) \right)^2 -\mu + h^2 W(x)\right) \gamma \right] - T\,
  S(\Gamma) \\ & \quad \ + \int_{\calC\times \R^d} V(h^{-1}(x-y))
  |\alpha(x,y)|^2 \, { dx \, dy} \label{def:bcs}
\end{align}
where $\calC$ denotes the unit cube $[0,1]^d$, and $\Tr$ stands for
the {\em trace per unit volume}. More precisely, if $B$ is a periodic
operator (i.e., it commutes with translations by $1$ in the $d$
coordinate directions), then $\Tr B$ equals, by definition the (usual)
trace of $\chi B$, with $\chi$ the characteristic function of
$\calC$. The location of the cube is obviously of no importance. It is
not difficult to see that the trace per unit volume has the usual
properties like cyclicity, and standard inequalities like H\"older's
inequality hold. This is discussed in more detail in
Section~\ref{sec:prel}.

We make the following assumptions on the functions $A$, $W$ and $V$
appearing in (\ref{def:bcs}).

\begin{assu}\label{as0}
  We assume both $W$ and $A$ to be periodic with period $1$. We
  further assume that $\widehat W(p)$ and $|\widehat A(p)|(1+|p|)$ are
  summable, with $\widehat W(p)$ and $\widehat A(p)$ denoting the
  Fourier coefficients of $W$ and $A$, respectively. In particular,
  $W$ is bounded and continuous and $A$ is in $C^1(\R^d)$.

  The interaction potential $V$ is assumed to be real-valued and
  reflection-symmetric, i.e., $V(x)=V(-x)$, with $V \in L^{p}(\R^d)$,
  where $p=1$ for $d=1$, $p>1$ for $d=2$ and $p=3/2$ for $d=3$.
\end{assu}

Our results presumably hold under slightly weaker regularity
assumptions on $W$ and $A$.  For the sake of transparency we shall not
aim for the weakest possible conditions, but rather try to keep the
proofs to a reasonable length.

The $L^p$ assumption on $V$ is the natural one to guarantee relative
form-bounded\-ness with respect to the Laplacian. A slower decay at
infinity could possibly be accommodated. Our method also works for
non-local potentials, which appear naturally in the theory of
superconductivity. For simplicity, we are not aiming for the most
general setting and work with Assumption~\ref{as0} from now on.

We note that a periodic magnetic field $B$, satisfying Maxwell's
equation $\nabla\cdot B =0$, can be described via a periodic vector
potential $A$ as $B=\nabla\wedge A$ if certain flux conditions are
satisfied. For $d=2$, $B$ has to have zero average, while for $d=3$
the flux through the boundaries of the unit cell has to vanish. This
follows from the fact that the relevant de Rham cohomology (closed
two-forms modulo exact ones) equals $\R^{n(n-1)/2}$ for the
$n$-dimensional torus.

\subsubsection{The Translation-Invariant Case}

In the translation invariant case, where $W=A=0$, it makes sense to
restrict $\F^{\rm BCS}$ to translation invariant states $\Gamma$. This
is the case studied in detail in \cite{HHSS}.\footnote{The results in
  \cite{HHSS} are worked out in three dimensions, but analogous
  results are easily seen to hold in one and two dimensions.} In
particular, it was shown in \cite{HHSS} that there is a critical
temperature $T_c\geq 0$ such that for $T\geq T_c$, the BCS functional
$\F^{\rm BCS}$ is minimized for $\alpha=0$ and $\gamma = (1 + \exp(
(-h^2\nabla^2-\mu)/T))^{-1}$, i.e., $\gamma$ is the one-particle
density matrix of a free Fermi gas. The critical temperature $T_c$ is
determined by the unique value of $T$ such that the operator
\begin{equation}\label{ktv}
  K_T + V(x)
\end{equation}
on $L^2_{\rm sym}(\R^d)$, the reflection-symmetric square-integrable
functions, has zero as its lowest eigenvalue.\footnote{In \cite{HHSS}
  the symmetry constraint on $\alpha$ was not explicitly enforced. The
  results hold equally if one works only in the subspace of reflection
  symmetric functions in $L^2(\R^d)$, however.} Here, $K_T$ denotes
the operator
\begin{equation}\label{def:kt}
  K_T = \frac{-\nabla^2 - \mu}{\tanh\left(\frac{-\nabla^2-\mu}{2T}\right)} \,.
\end{equation}
Note that $K_T$ is monotone increasing in $T$ and $K_T\geq 2T$. (If
$K_0 + V = |-\nabla^2 - \mu| + V \geq 0$, then $T_c = 0$.) In
particular, the essential spectrum of $K_T+V$ is $[2T,\infty)$, and
hence an eigenvalue at $0$ is necessarily isolated.

In the following, we shall assume that $T_c>0$ and that the ground
state of (\ref{ktv}) at $T=T_c$ is non-degenerate. We emphasize that
$T_c$ is independent of $h$.

\begin{assu}\label{as1}
  The potential $V$ is such that $T_c > 0$ and that $K_{T_c} + V$ has
  a non-degenerate ground state eigenvalue $0$.
\end{assu}

According to \cite{HS} this assumption is satisfied if $\widehat V
\leq 0$ (and not identically zero), for instance.  In the case that
$V$ is invariant under rotations, the non-degeneracy assumption means
that the minimizing function has angular momentum zero.  We denote the
eigenfunction of $K_{T_c}+V$ corresponding to eigenvalue zero by
$\alpha_0$.

\subsection{The GL Functional}

Let $\psi\in H^1_{\rm per}(\R^d)$, the periodic functions in $H_{\rm
  loc}^1(\R^d)$. For $\mathbb{B}_1$ a positive $d\times d$ matrix,
$B_2\in \R$ and $B_3>0$, the GL functional is defined as
\begin{align}\nonumber \label{GLfunct}
  \E^{\rm GL}(\psi) & = \int_{\calC} \left[ \big(
    \overline{\left(-i\nabla + 2 A(x)\right)\psi} \big)\cdot
    \mathbb{B}_1 \big( \left( -i\nabla + 2 A(x)\right)\psi\big)
  \right. \\ & \qquad\quad \left.  + B_2 W(x) |\psi(x)|^2 + B_3
    \left|1 - |\psi(x)|^2\right|^2\right] dx\,.
\end{align}
Note the coefficient $2$ in front of the vector potential $A$. It is
due to the fact that $\psi$ describes pairs of particles, and the
charge of a pair is twice the particle charge. The relevant
coefficients $\mathbb{B}_1$, $B_2$ and $B_3$ will be calculated below
from the BCS theory, see Eqs.~(\ref{def:b1})--(\ref{def:b3}). In the
rotation invariant case (i.e., for radial potentials $V$), the matrix
$\mathbb{B}_1$ is a multiple of the identity matrix.

The ground state energy of the GL functional will be denoted by
\begin{equation}
  E^{\rm GL} = \inf \left\{  \E(\psi)\, : \, \psi\in H^1_{\rm per}(\R^d)\right\}\,.
\end{equation}
It is not difficult to show that under our assumptions on $A$ and $W$,
there exists a corresponding minimizer, which satisfies a second order
differential equation known as the GL equation. The mathematical
aspects of the GL functional have been studied extensively in the
literature; see \cite{FH,Serfaty} and references therein.

\subsection{Main Results}\label{ss:results}

Recall the definition of the BCS functional $\F^{\rm BCS}$ in
(\ref{def:bcs}). Recall also that admissible states $\Gamma$ are of
the form (\ref{def:gamma}), are periodic with period 1 and with $0\leq
\gamma\leq 1$.  We define the energy $F^{\rm BCS}(T,\mu)$ as the
difference between the infimum of $\F^{\rm BCS}$ over all admissible
$\Gamma$ and the free energy of the normal state
\begin{equation}\label{def:gamma0}
  \Gamma_0 := \left( \begin{array}{cc} \gamma_0 & 0 \\ 0 & 1 -\bar\gamma_0 \end{array}\right) 
\end{equation}
where $\gamma_0 = (1+e^{((-ih\nabla+hA(x))^2 + h^2 W(x) -
  \mu)/T})^{-1}$. I.e.,
\begin{equation}\label{def:Fbcs}
  F^{\rm BCS}(T,\mu) = \inf_{\Gamma} \F^{\rm BCS}(\Gamma) - \F^{\rm BCS}(\Gamma_0) \,.
\end{equation}
Note that $\Gamma_0$ is the minimizer of the BCS functional in the
absence of interactions, i.e., when $V=0$. We have
\begin{equation}\label{f0}
  \F^{\rm BCS}(\Gamma_0) = - T\, \Tr \ln\left( 1+ \exp\left(-\left( \left(-i h \nabla + h A(x) \right)^2 -\mu + h^2 W(x)  \right) /T \right) \right) \,,
\end{equation}
which is $O(h^{-d})$ for small $h$.  The system is said to be in a
superconducting (or superfluid, depending on the physical
interpretation) state if $F^{\rm BCS} < 0$.

\begin{Theorem}\label{thm:main}
  As above, let $T_c>0$ denote the critical temperature in the
  translation invariant case, and assume Assumptions~\ref{as0}
  and~\ref{as1}. Let $D>0$. For appropriate coefficients
  $\mathbb{B}_1$, $B_2$ and $B_3$ (given in
  (\ref{def:b1})--(\ref{def:b3}) below), we have, as $h\to 0$,
  \begin{equation}\label{enthm}
    F^{\rm BCS}(T_c(1-h^2D),\mu) =  h^{4-d} \left(E^{\rm GL} - B_3 + e \right) \,,
  \end{equation}
  with $e$ satisfying the bounds
  \begin{equation}
    \const h \geq e \geq  - e_{\rm L} : = - \const   \times \left\{ \begin{array}{cl} h^{1/3} & \text{for $d=1$} \\ h^{1/3} \left[\ln (1/h)\right]^{1/6} & \text{for $d=2$} \\ h^{1/5} & \text{for $d=3$.} \end{array}\right.
  \end{equation} 

  Moreover, if $\Gamma$ is an approximate minimizer of $\F^{\rm BCS}$
  at $T=T_c(1-h^2 D)$, in the sense that
  \begin{equation}
    \F^{\rm BCS}(\Gamma)\leq \F^{\rm BCS}(\Gamma_0) + h^{4-d} \left( E^{\rm GL} - B_3 + \epsilon\right)
  \end{equation}
  for some small $\epsilon > 0$, then the corresponding $\alpha$ can
  be decomposed as
  \begin{equation}\label{thm:dec}
    \alpha  = \frac h2 \big( \psi( x) \widehat\alpha_0(-ih\nabla) + \widehat \alpha_0(-ih\nabla) \psi(x)\big)  + \sigma
  \end{equation}
  with $\E^{\rm GL}(\psi) \leq E^{\rm GL} + \epsilon + e_{\rm L}$,
  $\alpha_0$ the (appropriately normalized; see (\ref{normal}))
  zero-energy ground state of (\ref{ktv}) at $T=T_c$, and
  \begin{equation}\label{sigmab}
    \int_{\calC\times \R^d} |\sigma(x,y)|^2 \, dx\, dy  \leq  \const \frac {h^{4-d}}{e_{\rm L}^2}\,.
  \end{equation}
\end{Theorem}

To appreciate the bound (\ref{sigmab}), note that the square of the
$L^2(\calC\times \R^d)$ norm of the first term on the right side of
(\ref{thm:dec}) is of the order $h^{2-d}$, and hence is much larger
than the one of $\sigma$ for small $h$. Theorem~\ref{thm:main} thus
justifies the claim made in the Introduction in (\ref{alphconverg}).

For smooth enough $A$ and $W$, one could also expand $\F^{\rm
  BCS}(\Gamma_0)$ in (\ref{f0}) to order $h^{4-d}$ and thus obtain the
total energy $\inf_\Gamma \F^{\rm BCS}(\Gamma)$ to this order. Our
approach bounds directly the difference (\ref{def:Fbcs}), however, and
hence it is not necessary to compute $\F^{\rm BCS}(\Gamma_0)$ in
detail, and to make the corresponding additional regularity
assumptions on $A$ and $W$.

Our proof shows that the coefficients $\mathbb{B}_1$, $B_2$ and $B_3$
in the GL functional are given as follows. Recall that $\alpha_0$
denotes the unique ground state of $K_{T_c} + V$. It satisfies
$\alpha_0(x) = \alpha_0(-x)$, and we can take it to be real. Let $t$
denote the Fourier transform of $2 K_{T_c} \alpha_0 = -2 V \alpha_0$,
i.e.,
\begin{equation}\label{deft}
  t(p) = - 2 (2\pi)^{-d/2} \int_{\R^d} V(x) \alpha_0(x) e^{-ip\cdot x} dx \,.
\end{equation}
We normalize it, i.e., choose the normalization of $\alpha_0$, such
that
\begin{equation} \label{normal} \int_{\R^d} t(q)^4
  \,\frac{g_1(\beta_c(q^2-\mu))}{q^2-\mu}\, dq = \frac D {\beta_c}
  \int_{\R^d} t(q)^2
  \cosh^{-2}\left(\tfrac{\beta_c}{2}(q^2-\mu)\right) \, dq\,,
\end{equation}
where $\beta_c = 1/T_c$, $D = (1-T/T_c)/h^2$ as above, and $g_1$
denotes the function
$$
g_1(z) = \frac{ e^{2 z} - 2 z e^{z}-1}{z^2 (1+e^{z})^2}\,.
$$
Note that $g_1(z)/z > 0$.  Define also
$$
g_2(z) = \frac{2 e^{z} \left( e^{ z}-1\right)}{z
  \left(e^{z}+1\right)^3}
$$
(compare with (\ref{deff})--(\ref{defg2})).  Then $\mathbb{B}_1$ is
the matrix with components
\begin{equation}\label{def:b1}
  \left(\mathbb{B}_1\right)_{ij} =  \frac{ \beta_c^2}{16}   \int_{\R^d} t(q)^2 \left( \delta_{ij} g_1(\beta_c(q^2-\mu)) + 2 \beta_c q_i q_j\, g_2(\beta_c(q^2-\mu)) \right) \frac{dq}{(2\pi)^d}\,.
\end{equation}
Note that in case $V$ is radial, also $\alpha_0$ and $t$ are radial
and thus $B_1$ is proportional to the identity matrix. Moreover, $B_2$
and $B_3$ are given by
\begin{equation}\label{def:b2}
  B_2 =  \frac  {\beta_c^2 }4 \int_{\R^d} t(q)^2 \, g_1(\beta_c(q^2-\mu)) \, \frac{dq}{(2\pi)^d}  
\end{equation}
and
\begin{equation}\label{def:b3}
  B_3 =   \frac {\beta_c^2} {16}  \int_{\R^d} t(q)^4 \, \frac{g_1(\beta_c(q^2-\mu))}{q^2-\mu}\,\frac{dq}{(2\pi)^d} \,,
\end{equation}
respectively.  Alternatively, $B_3$ could be written using the
normalization (\ref{normal}). In particular, $B_3/|B_2|$ and
$B_3/\|\mathbb{B}_1\|$ are proportional to $D$, and hence proportional
to the difference of the temperature to the critical one.

Note that $B_3>0$ since $g_1(z)/z > 0$ for all $z\in \R$. The
coefficient $B_2$ can, in principle, have either sign if $\mu > 0$,
however. It has the same sign as the derivative of $T_c$ with respect
to $\mu$. This is not surprising. The external potential plays the
role of a local variation in the chemical potential. If increasing
$\mu$ increases $T_c$, the pairing mechanism should be enhanced at
negative values of $W$ and the density of Cooper pairs should be
largest there. If increasing $\mu$ decreases $T_c$, however, the
situation is reversed.

To see that the matrix $\mathbb{B}_1$ is positive, we calculate
\begin{align}\nn
  & 2 \re \,\langle \alpha_0| x_i (K_{T_c}+V) x_j|\alpha_0\rangle =
  \langle \alpha_0 | [x_i,[K_{T_c},x_j]]|\alpha_0\rangle \\ &= 8
  (2\pi)^d \mathbb{B}_{1,ij} - 2 {\beta_c^3} \int_{\R^d} t(q)^2\, q_i
  q_j \, \frac{ g_1 (\beta_c(q^2-\mu)) }{\sinh(\beta_c(p^2-\mu) } \,
  dq \,.
\end{align}
The last term on the right side defines a negative matrix, while the
left side defines a positive matrix, since $K_{T_c}+V\geq 0$. Hence
$\mathbb{B}_1$ is positive.

\subsection{Outline of the paper}

The structure of the remainder of this paper is as follows. In
Section~\ref{sec:semi} we shall state our main semiclassical
estimates. These are a crucial input to obtain the bounds in
Theorem~\ref{thm:main}.  The leading terms in our semiclassical
expansion can, in principle, be obtained from well-known formulas in
semiclassical analysis, but the standard techniques do not apply
directly because we are forced to work with rather minimal regularity
assumptions. We shall formulate the main results in separate theorems;
see Theorems~\ref{thm:scl} and~\ref{lem3} below. Their proofs, which
are rather lengthy, will be given in Sections~\ref{sec:thm:scl}
and~\ref{sec:propproof}, respectively. Some technicalities are
deferred to the appendix.

Section~\ref{sec:prel} explains various inequalities for the trace per
unit volume which will be used throughout the proofs.  Properties of
$\alpha_0$, the ground state of $(\ref{ktv})$ at $T=T_c$, are derived
in Section~\ref{sec:alpha0}. It is shown that $t$, defined in
(\ref{deft}) above, is smooth and has a suitable decay at
infinity. The results imply, in particular, that the coefficients
(\ref{def:b1})--(\ref{def:b3}) are well-defined and finite.

An upper bound on $F^{\rm BCS}$ will be derived in
Section~\ref{sec:up}, using the variational principle. An important
input will be the semiclassical estimates of Theorems~\ref{thm:scl}
and~\ref{lem3}.

Sections~\ref{sec:low1} and~\ref{sec:low2} contain the lower bound.
In the first part, the structure of an approximate minimizer is
investigated, which leads to a definition of the order parameter
$\psi$. This structure is then a crucial input to the second part,
where also the semiclassical estimates of Section~\ref{sec:semi}
enter.

Throughout the proofs, $C$ will denote various different constants. We
will sometimes be sloppy and use $C$ also for expressions that depend
only on some fixed, $h$-independent, quantities like $\mu$, $T_c$ or
$\|W\|_\infty$, for instance.

\section{Semiclassical Estimates}\label{sec:semi}

One of the key ingredients in the proof of Theorem~\ref{thm:main} is
semiclassical analysis.  Choose a periodic function $\psi\in H^2_{\rm
  loc}(\R^d)$ and a sufficiently nice function $t$ and let $\Delta$ be
the operator
\begin{equation}\label{def:delta}
  \Delta = - \frac h 2\left(\psi(x) t(-ih\nabla) + t(-ih\nabla) \psi(x)\right) \,. 
\end{equation}
It has the integral kernel
\begin{equation}
  \Delta(x,y) = -\frac{h^{1-d}}{2(2\pi)^{d/2}} \left( \psi( x) + \psi( y) \right) \widehat t(h^{-1}(y-x)) \,.
\end{equation}
Our convention for the Fourier transform is that $\widehat f(p) =
(2\pi)^{-d/2} \int_{\R^d} f(x) e^{-ip\cdot x} dx$.  We shall assume
that
\begin{equation}\label{t:as1}
  \partial^\gamma t \in L^{2p/(p-1)}(\R^d) 
\end{equation}
(with $p$ defined in Assumption~\ref{as0}) and that
\begin{equation}\label{t:as2}
  \int_{\R^d} \frac{|\partial^\gamma t(q)|^2}{ 1+q^2}\, dq < \infty
\end{equation}
for all $\gamma \in \{0,1\,\dots,4\}^d$.  For simplicity, we also
assume that $t$ is reflection-symmetric and real-valued. For the
function $t$ in (\ref{deft}), these assumptions are satisfied, as will
be shown in Section~\ref{sec:alpha0}.

Let $H_\Delta$ be the operator
\begin{equation}\label{hdelta}
  H_\Delta = \left( \begin{array}{cc}  \left(-i h \nabla + h A(x) \right)^2 -\mu + h^2 W(x) & \Delta \\ \bar\Delta & - \left(i h \nabla + h A(x) \right)^2 +\mu - h^2 W(x) \end{array} \right)
\end{equation}
on $L^2(\R^d)\otimes \C^2$, with $A$ and $W$ satisfying
Assumption~\ref{as0}.  In the following, we will investigate the trace
per unit volume of functions of $H_\Delta$. Specifically, we are
interested in the effect of the off-diagonal term $\Delta$ in
$H_\Delta$, in the semiclassical regime of small $h$.

\begin{Theorem}\label{thm:scl}
  Let
  \begin{equation}\label{deff}
    f(z) = -   \ln \left(1+ e^{-z}\right) \ ,
  \end{equation}
  and define
  \begin{equation}\label{defg0}
    g_0(z) = \frac{f'(-z) - f'(z)}{z} = \frac { \tanh\left(\tfrac 12 z\right)}{z}\,,
  \end{equation}
  \begin{equation}\label{defg1}
    g_1(z) = - g_0'(z) = \frac{f'(-z)-f'(z)}{z^2} + \frac{f''(-z)+f''(z)}{z} = \frac{ e^{2 z} - 2 z e^{z}-1}{z^2 (1+e^{z})^2}
  \end{equation}
  and
  \begin{equation}\label{defg2}
    g_2(z) =  g_1'(z) + \frac 2 z \,g_1(z) =  \frac{f'''(z)-f'''(-z)}{z} = \frac{2 e^{z} \left( e^{ z}-1\right)}{z \left(e^{z}+1\right)^3}\,.
  \end{equation}
  Then, for any $\beta > 0$, the diagonal entries of the $2\times 2$
  matrix-valued operator $f(\beta H_\Delta) - f(\beta H_0)$ are
  locally trace class, and the sum of their traces per unit volume
  equals
  \begin{align}\nn
    \frac {h^{d}}{\beta}\, \Tr\left[ f(\beta H_\Delta) - f(\beta
      H_0)\right] & = h^2 E_1 + h^4 E_2 + O(h^{5}) \left(
      \|\psi\|^4_{H^1(\calC)} + \|\psi\|^2_{H^1(\calC)} \right) \\
    &\quad + O(h^6) \left( \|\psi\|_{H^1(\calC)}^6+
      \|\psi\|_{H^2(\calC)}^2\right)\,, \label{210}
  \end{align}
  where
  \begin{equation}\label{def:e1}
    E_1 =  -\frac {\beta} 2 \|\psi\|_2^2  
    \int_{\R^d} t(q)^2  \, g_0(\beta(q^2-\mu))\,  \frac{dq}{(2\pi)^d} 
  \end{equation}
  and
  \begin{align}\nn
    E_2 &= - \frac { \beta} 8 \sum_{j,k=1}^d \langle \partial_j
    \psi| \partial_k \psi\rangle \int_{\R^d} t(q)
    \left[ \partial_j \partial_k t\right]\!(q)\, g_0(\beta(q^2-\mu))\,
    \frac{dq}{(2\pi)^d} \\ \nonumber & \quad + \frac{ \beta^2}8
    \sum_{j,k=1}^d \langle (\partial_j + 2 i A_j) \psi| (\partial_k +
    2 iA_k) \psi \rangle \\ \nonumber & \qquad \qquad \quad \times \!
    \int_{\R^d} t(q)^2 \left( \delta_{jk} g_1(\beta(q^2-\mu)) + 2
      \beta
      q_j q_k\, g_2(\beta(q^2-\mu)) \right) \frac{dq}{(2\pi)^d} \\
    \nonumber & \quad + \frac {\beta^2 }2 \langle \psi|
    W|\psi\rangle\int_{\R^d} t(q)^2 \, g_1(\beta(q^2-\mu)) \,
    \frac{dq}{(2\pi)^d} \\ & \quad + \frac {\beta^2} 8 \|\psi\|_4^4
    \int_{\R^d} t(q)^4 \, \frac{g_1(\beta(q^2-\mu))}{q^2-\mu}\,
    \frac{dq}{(2\pi)^d} \,. \label{def:e2}
  \end{align}
  The error terms in (\ref{210}) of order $h^5$ and $h^6$ are bounded
  uniformly in $\beta$ for $\beta$ in compact intervals in
  $(0,\infty)$. They depend on $t$ only via bounds on the expressions
  (\ref{t:as1}) and (\ref{t:as2}).
\end{Theorem}

Here, we use the short-hand notation $\|\psi\|_p$ for the norm on
$L^p(\calC)$. Likewise, $\langle \,\cdot\, | \, \cdot\, \rangle$
denotes the inner product on $L^2(\calC)$.

In general, the operator $f(\beta H_\Delta) - f(\beta H_0)$ will not
be trace class under our assumptions on $t$ and $\psi$. Hence the
trace in (\ref{210}) has to be suitably understood as the sum of the
traces of the diagonal entries. This issue is further discussed in
Section~\ref{sec:up}.

The expressions $E_1$ and $E_2$ are the first two non-vanishing terms
in a semiclassical expansion of the left side of (\ref{210}). They can
be obtained, in principle, from well-known formulas in semiclassical
analysis \cite{helffer,robert}. The standard techniques are not
directly applicable in our case, however. This has to do, on the one
hand, with our rather minimal regularity assumptions on $W$, $A$,
$\psi$ and $t$ and, on the other hand, with the fact that we are
working with the trace per unit volume of an infinite, periodic
system.

The proof of Theorem~\ref{thm:scl} will be given in
Section~\ref{sec:thm:scl}. As the proof shows, the theorem holds for a
larger class of functions $f$, satisfying appropriate smoothness and
decay assumptions.

Our second semiclassical estimate concerns the upper off-diagonal term
of
\begin{equation}\label{def:gammad}
  \Gamma_\Delta = \frac{ 1}{1+ e^{\beta H_\Delta}}\,,
\end{equation}
which we shall denote by $\alpha_\Delta$. We shall be interested in
its $H^1$ norm. In general, we define the $H^1$ norm of a periodic
operator $O$ by
\begin{equation}\label{def:h1}
  \|O\|_{H^1}^2 = \Tr \left[ O^\dagger \left(1-h^2\nabla^2\right) O  \right]\,.
\end{equation}
In other words, $\|O\|_{H^1}^2 = \|O\|_2^2 + h^2 \|\nabla
O\|_2^2$. Note that this definition is not symmetric, i.e.,
$\|O\|_{H^1} \neq \|O^\dagger\|_{H^1}$ in general.

Given $t$ as above, let us define the function
\begin{equation}\label{def:varphi}
  \varphi(p) = \frac \beta 2 \,g_0(\beta (p^2-\mu)) \,t(p) \,.
\end{equation}

\begin{Theorem}\label{lem3} With $\alpha_\Delta$ denoting the upper
  off-diagonal entry of $\Gamma_\Delta$ in (\ref{def:gammad}), we have
  \begin{equation}
    \left\| \alpha_\Delta - \tfrac h2 \left(\psi(x) \varphi(-ih\nabla) + \varphi(-ih\nabla) \psi(x) \right) \right\|_{H^1}  \leq C h^{3-d/2}  \|\psi\|_{H^2(\calC)} \,.
  \end{equation}
\end{Theorem}

The proof of this theorem will be given in
Section~\ref{sec:propproof}.

\section{Preliminaries}\label{sec:prel}

In this section we collect a few useful facts about the trace per unit
volume. In particular, we recall the general form of H\"older's
inequality and Klein's inequality, both of which will be used several
times in the proofs below.

Let $A$ and $B$ be bounded periodic operators on either $L^2(\R^d)$ or
$L^2(\R^d;\C^2)$, i.e., operators that commute with translations by a
unit length in any of the $d$ coordinate directions. The trace per
unit volume of $A$ is simply defined as the trace of $\chi A \chi$,
where $\chi$ is the characteristic functions of a unit cube, i.e., the
projection onto functions supported in this cube. Obviously, the
location of the cube is irrelevant. For $p\geq 1$ we also denote the
$p$-norm of $A$ by
\begin{equation}\label{def:pnorm}
  \|A\|_p = \left( \Tr \left(A^\dagger A\right)^{p/2} \right)^{1/p} \,.
\end{equation}
Here and in the remainder of this paper, $\Tr $ denotes the trace per
unit volume. We also use the notation $\|A\|_\infty$ for the standard
operator norm.

For any $1\leq p\leq \infty$, the triangle inequality
\begin{equation}\label{triangle}
  \| A + B \|_p\leq \|A\|_p + \|B\|_p
\end{equation}
holds.  For $1/r+1/s=1/p$, $1\leq r,s,p\leq \infty$, we have the
general H\"older inequality
\begin{equation}\label{holder}
  \|A B\|_p \leq \|A\|_r \|B\|_s \,.
\end{equation}
Moreover, if $f:\R\times \R \to \R_+$ is a non-negative function of
the form $f(x,y) = \sum_{i} g_i(x) h_i(y)$, then Klein's inequality
\begin{equation}\label{klein}
  \Tr f(A,B) : = \sum_i \Tr g_i(A) h_i(B)  \geq 0
\end{equation}
holds for self-adjoint $A$ and $B$.

Inequalities (\ref{triangle})--(\ref{klein}) are well-known in the
case of standard traces, see \cite{simon} and \cite{thirring}. They
extend to the periodic case via the Floquet decomposition
\cite[Sect.~XIII.16]{ReSi}. Specifically,
\begin{equation}
  A \cong \int^\oplus_{[0,2\pi]^d} A^\xi\, \frac{d\xi}{(2\pi)^d}
\end{equation}
with $A^\xi$ operating on $L^2(\calC)$, and $\cong$ denotes unitary
equivalence. The trace per unit volume equals
\begin{equation}
  \Tr A = \int_{[0,2\pi]^d} \Tr_{L^2(\calC)} A^\xi \, \frac{d\xi}{(2\pi)^d}
\end{equation}
The inequalities (\ref{triangle})--(\ref{klein}) then easily follow
from the standard ones using that
\begin{equation}
  (AB)^\xi = A^\xi B^\xi
\end{equation}
and that $g(A)^\xi = g(A^\xi)$ for appropriate functions $g$ and
self-adjoint $A$.

One also checks that
\begin{equation}
  \left| \Tr A \right| \leq \|A\|_1 \,.
\end{equation}
By induction, if follows from (\ref{holder}) that
\begin{equation}\label{genholder}
  \left| \Tr A_1 A_2 \cdots A_n \right| \leq \left\| A_1 A_2 \cdots A_n \right\|_1 \leq \prod_{i=1}^n \|A_i\|_{p_i} \,,
\end{equation}
where $1\leq p_i \leq \infty$ and $\sum_i p_i^{-1} = 1$.

Note that while the local trace norms introduced here share many
properties with the usual Schatten norms, they are not monotone
decreasing in $p$. For instance, if $A^\xi$ is a rank one operator
(whose norm is not independent of $\xi$) then the norm $\|A\|_p$ is
actually {\em increasing} in $p$. In particular, the finiteness of the
$p$-norm does not, in general, imply finiteness of the $q$-norm for
$q>p$.

In the proofs below we need one more inequality that generalizes an
inequality by Lieb and Thirring to the case of the trace per unit
volume (see, e.g., \cite[Sect.~4.5]{LSstab}. If $A$ and $B$ are
periodic operators and $p\geq 2$, then
\begin{equation}\label{alt}
  \|AB\|_p^p =  \Tr |AB|^p = \Tr \left( |B^\dagger| |A|^2 |B^\dagger| \right)^{p/2} \leq \Tr |B^\dagger|^{p/2} |A|^p |B^\dagger|^{p/2} \,.
\end{equation}
Again, the proof follows from the standard one using the Floquet
decomposition.  In the special case that $A$ is a multiplication by a
periodic function $a(x)$ and $B$ is multiplication in Fourier space,
i.e., $B = b(-i\nabla)$, (\ref{alt}) reads
\begin{equation}\label{lte}
  \| a(x) b(-i\nabla) \|_p \leq  (2\pi)^{-d/p} \|a\|_{L^p(\calC)} \|b\|_{L^p(\R^d)}\,.
\end{equation}

\section{Properties of $\alpha_0$}\label{sec:alpha0}

Recall the definition of $K_T$ in (\ref{def:kt}), and recall that
$\alpha_0$ denotes the eigenfunction of $K_{T_c}+V$ corresponding to
the eigenvalue zero, which we assume to be simple. We assume that
$T_c>0$, and that $V \in L^{p}(\R^d)$, with $p=1$ for $d=1$, $p>1$ for
$d=2$ and $p=3/2$ for $d=3$. There is no parameter $h$ in the
definition of $K_T$; in other words, $h=1$ in this section.

\begin{prop} \label{prop1}
  \begin{equation}
    e^{\kappa |x|} \sqrt{|V|} \alpha_0 \in L^2(\R^d)
  \end{equation}
  for any $\kappa < \kappa_c:=\im \sqrt{\mu + i \pi T_c}$.
\end{prop}

In three dimensions, this proposition also holds for $\kappa =
\kappa_c$, in fact.

The key to prove Proposition~\ref{prop1} is to estimate the integral
kernel of $K_T^{-1}$. It turns out to have the same behavior as the
one for $(-\nabla^2+\kappa_c^2)^{-1}$, both at the origin and at
infinity. Hence we find it convenient to state our bound in terms of
$G(x-y;\lambda)$, the Green's function of $-\nabla^2-\lambda$ for
$\lambda\in \C \setminus \R_+$. Note that $G(x,\lambda)>0$ for
$\lambda <0$.

\begin{Lemma}\label{lem:expdecay}
  For some constant $C_{\mu/T}$ depending on $\mu/T$, the integral
  kernel of $K_T^{-1}$ satisfies
  \begin{equation}
    \left| K_T^{-1} (x-y) \right| \leq  C_{\mu/T}  G(x-y;- (\im \sqrt{\mu + i \pi T})^2)\,.
  \end{equation}
\end{Lemma}

\begin{proof}
  We start with the series representation
  \begin{equation}
    \frac {\tanh{z}}{z}  = \sum_{n = 1}^{\infty} \frac{ 2}{ (n-\tfrac 12)^2 \pi^2 + z^2}
  \end{equation}
  which implies that
  \begin{equation}
    K_T^{-1} =  \sum_{n=1}^\infty   \frac { 4 T}{  (n-\tfrac 12)^2 (2\pi T)^2  +  (-\nabla^2-\mu)^2} \,.
  \end{equation}
  We can rewrite this as
  \begin{equation}
    K_T^{-1} =  \frac 2\pi \sum_{n=1}^\infty  \frac 1{n-\tfrac 12} \im   \frac { 1 }{ -\nabla^2-\mu - i (n-\tfrac 12) 2\pi T} \,.
  \end{equation}
  Hence
  \begin{equation}
    K_T^{-1}(x-y) =  \frac 2\pi \sum_{n=1}^\infty  \frac 1{n-\tfrac 12} \im  G(x-y; \mu + i (n-\tfrac 12) 2\pi T)\,.
  \end{equation}
  The result now follows easily using the explicit behavior of the
  Green's function $G$ at $0$ and at infinity.
\end{proof}

\begin{Lemma}\label{lem:intop}
  Let $L$ be an integral operator with integral kernel bounded as
  \begin{equation}
    | L(x,y) | \leq A e^{-\kappa|x-y|} G(x-y;-e)
  \end{equation}
  for $\kappa>0$ and $e> 0$ (or $e\geq 0$ in $d=3$). Then there is a
  finite constant $C_{p,d}$ such that for any $U_1,U_2 \in
  L^{2p}(\R^d)$
  \begin{equation}
    \left\| e^{\kappa|x|} U_1 L U_2 e^{-\kappa|x|} \right\| \leq A\, C_{p,d}\, e^{d/(2p)-1}  \|U_1\|_{2p} \|U_2\|_{2p} 
  \end{equation}
  where $p\geq 1$ in $d=1$, $p>1$ in $d=2$ and $p\geq 3/2$ in $d=3$.
\end{Lemma}

\begin{proof}
  For $\psi\in L^2(\R^d)$,
  \begin{align}\nn
    &\left\| e^{\kappa|x|} U_1 L U_2 e^{-\kappa|x|} \psi \right\|_2^2
    \\ \nn & \leq \int_{\R^{3d}} |\psi(y)| e^{-\kappa|y|} |U_2(y)|
    |L(x,y)| |U_1(x)|^2 e^{2\kappa|x|} \\ \nn & \qquad\qquad \times
    |L(x,y')| |U_2(y')| e^{-\kappa|y'|} |\psi(y')| \, dy\, dy'\, dx \\
    \nn & \leq A^2 \int_{\R^{3d}} |\psi(y)| |U_2(y)| G(x-y;-e)
    |U_1(x)|^2 G(x-y';-e) |U_2(y')| |\psi(y')| \, dy\, dy'\, dx \\
    &=A^2 \left\| |U_1| (-\nabla^2+e)^{-1} |U_2| |\psi|
    \right\|_2^2\,,
  \end{align}
  where we used the triangle inequality $2|x| - |y|- |y'| \leq
  |x-y|+|x-y'|$. The result now follows from the fact that the
  operator norm of $|U_1|(-\nabla^2+e)^{-1}|U_2|$ is bounded by a
  constant times $e^{d/(2p)-1}\|U_1\|_{2p} \|U_2\|_{2p}$, as can be
  seen by an application of the Hausdorff-Young inequality (for $e>0$)
  or the Hardy-Littlewood-Sobolev inequality (for $e=0$ and $d=3$) and
  H\"older's inequality.
\end{proof}

\begin{proof}[Proof of Proposition~\ref{prop1}]
  The function $\phi = \sqrt{|V|} \alpha_0$ is in $L^2(\R^d)$ (because
  of the relative form-boundedness of $V$) and satisfies
  \begin{equation}
    \phi = - \sqrt{|V|} \frac 1{K_{T_c}} \sqrt{V} \phi \,,
  \end{equation}
  where $\sqrt{V}:= V/\sqrt{|V|}$. For given $R>0$, we decompose $\phi
  = \phi_1 + \phi_2$, where $\phi_2 = \phi\,
  \chi_{\{|x|>R\}}$. Introducing $U_1 = - \chi_{\{|x|>R\}} \sqrt{|V|}$
  and $U_2 = \chi_{\{|x|>R\}}\sqrt{V}$, we have
  \begin{equation}
    \phi_2 = U_1 \frac 1{K_{T_c}} U_2 \phi_2 + f \,, 
  \end{equation}
  where
  \begin{equation}
    f = U_1 \frac 1{K_{T_c}} \sqrt{V} \phi_1 \,.
  \end{equation}

  We shall now use that the Green's function of the Laplacian has the
  property
  \begin{equation}
    G(x-y;-\kappa^2) \leq C_{\epsilon/\kappa} e^{-\left(\kappa-e\right)|x-y|} G(x-y;-e^2)
  \end{equation}
  for $0<e\leq \kappa$. In combination with Lemma~\ref{lem:expdecay}
  and Lemma~\ref{lem:intop}, this implies that for $\kappa<\kappa_c$,
  the operator
  \begin{equation}
    e^{\kappa|x|} U_1 \frac 1{K_{T_c}} \sqrt{V} e^{-\kappa|x|}
  \end{equation}
  is bounded. Since $e^{\kappa|x|}|\phi_1|\leq e^{\kappa R}|\phi_1|$,
  we conclude that $e^{\kappa|x|} f \in L^2(\R^d)$. Similarly, we
  observe that
  \begin{equation}
    \left\| e^{\kappa|x|} U_1 \frac 1{K_{T_c}} U_2 e^{-\kappa|x|} \right\| \leq C \| V\, \chi_{\{|x|>R\}}\|_p \,,
  \end{equation}
  with $p$ as in Assumption~\ref{as0}. The latter expression is less
  than one for $R$ large enough. Hence
  \begin{equation}
    e^{\kappa|x|} \phi_2 = \left( 1 - e^{\kappa|x|} U_1 \frac 1{K_{T_c}} U_2 e^{-\kappa|x|}\right)^{-1} e^{\kappa|x|} f  
  \end{equation}
  is an element of $ L^2(\R^d)$.  Since obviously also
  $e^{\kappa|x|}\phi_1 \in L^2(\R^d)$, this concludes the proof.
\end{proof}

We shall use the result of Proposition~\ref{prop1} in the following
way. Recall that $t$ was defined in (\ref{deft}) as twice the Fourier
transform of $K_{T_c} \alpha_0 = - V \alpha_0$. The following
proposition collects all the regularity properties of $t$ and
$\alpha_0$ that we shall need below.

\begin{prop} \label{prop:reg} The function $t$ in (\ref{deft}),
  together with all its derivatives, is a function in $L^q(\R^d)$,
  with $q = 2p/(p-1)$, i.e., $q=\infty$ for $d=1$, $2<q<\infty$ for
  $d=2$ and $q=6$ for $d=3$.  We also have that
  \begin{equation}\label{propt2}
    \int_{\R^d} \frac{|\partial^\gamma t(p)|^2}{1+p^2} dp = 4 \left\langle x^\gamma \sqrt{|V|}\alpha_0 \left| \sqrt{V} \frac 1{ 1-\nabla^2} \sqrt{V} \right| x^\gamma \sqrt{|V|} \alpha_0\right\rangle < \infty
  \end{equation}
  and that $\int_{\R^d} \left( | x^\gamma \nabla \alpha_0(x)|^2 +
    |x^\gamma \alpha_0(x)|^2\right) dx <\infty$ for all $\gamma \in
  \N_0^d$.
\end{prop}

\begin{proof}
  The $L^q$ property follows easily from the Hausdorff-Young
  \cite[Thm.~5.7]{LL} and H\"older inequality,
  \begin{equation}\label{propt1}
    \| \partial^\gamma t \|_{2p/{(p-1)}} \leq 2 \| x^\gamma V \alpha_0\|_{2p/(p+1)} \leq 2 \|V\|_p^{1/2} \| x^\gamma \sqrt{|V|} \alpha_0 \|_2\,,
  \end{equation}
  using Proposition~\ref{prop1} and our assumption that $V \in L^p$.

  The boundedness of (\ref{propt2}) follows again from
  Proposition~\ref{prop1} and the fact that the operator on the right
  side is bounded by our assumption on $V$.  Finally, the last
  statement follows easily from (\ref{propt2}) writing the integral in
  Fourier space and using that $\widehat \alpha_0(p) = t(p)/(2
  K_{T_c}(p))$, where $K_{T_c}(p)$ denotes the function obtained by
  replacing $-\nabla^2$ by $p^2$ in (\ref{def:kt}).
\end{proof}

For $d=3$, one could even drop the $1$ in $1-\nabla^2$ and $1+p^2$ in
(\ref{propt2}); boundedness of the operator follows from the
Hardy-Littlewood-Sobolev inequality.

\section{Proof of Theorem~\ref{thm:main}: Upper Bound}\label{sec:up}

In this section, we shall prove the upper bound in
Theorem~\ref{thm:main}, i.e., we shall show that $e\leq C h$ with $e$
defined in (\ref{enthm}). We shall denote $\beta = 1/T$ and
$\beta_c=1/T_c$. Recall that $\beta_c/\beta = 1 - h^2 D$ with $D>0$.

By definition, $F^{\rm BCS}(T,\mu)\leq \F^{\rm BCS}(\Gamma) - \F^{\rm
  BCS}(\Gamma_0)$ for any admissible state $\Gamma$.  As a trial
state, we use
\begin{equation}\label{def:gammad2}
  \Gamma_\Delta = \left( \begin{array}{cc} \gamma_\Delta & \alpha_\Delta  \\ \bar\alpha_\Delta & 1-\bar\gamma_\Delta \end{array} \right) = \frac 1{1 + e^{\beta H_\Delta}}
\end{equation}
where $H_\Delta$ is given in (\ref{hdelta}) with $\Delta$ as in
(\ref{def:delta}). Note that $\Gamma_0$, defined in
(\ref{def:gamma0}), indeed corresponds to setting $\Delta=0$ in
(\ref{def:gammad2}). For $t$, we choose (\ref{deft}), which is
reflection symmetric and can be taken to be real. We normalize it such
that (\ref{normal}) holds. The integral kernel of $\Delta$ is then
given by
\begin{equation}\label{int:delta}
  \Delta(x,y) = \frac{h^{1-d}}{(2\pi)^{d/2}} \left( \psi(x) + \psi(y) \right) V(h^{-1}(x-y)) \alpha_0(h^{-1}(x-y)) \,.
\end{equation}

Note that $H_\Delta$ is unitarily equivalent to $-\bar H_{\Delta}$,
\begin{equation}\label{eq:unit}
  U H_{\Delta} U^\dagger = - \bar H_\Delta \quad \text{with}\quad  U = \left( \begin{array}{cc} 0 & 1 \\ -1 & 0 \end{array} \right)\,.
\end{equation}
Hence also $U\Gamma_\Delta U^\dagger = 1 - \bar \Gamma_\Delta$ and, in
particular,
\begin{equation}\label{ent:ref}
  S(\Gamma_\Delta) = - \tfrac 12 \Tr \left[ \Gamma_\Delta \ln \Gamma_\Delta  +  (1-\Gamma_\Delta) \ln (1-\Gamma_\Delta)\right]\,.
\end{equation}
A simple calculation shows that
\begin{align}\nonumber
  & \Gamma_\Delta \ln \Gamma_\Delta + (1-\Gamma_\Delta) \ln
  (1-\Gamma_\Delta) - \Gamma_0 \ln \Gamma_0 - (1-\Gamma_0) \ln
  (1-\Gamma_0) \\ & = - \beta H_\Delta \Gamma_\Delta + \beta H_0
  \Gamma_0 - \ln \left( 1+e^{-\beta H_\Delta}\right) + \ln \left(
    1+e^{-\beta H_0}\right) \,. \label{tci1}
\end{align}
Moreover,
\begin{equation}\label{tci2}
  H_\Delta \Gamma_\Delta - H_0 \Gamma_0 = \left( \begin{array}{cc} k ( \gamma_\Delta - \gamma_0) + \Delta \bar\alpha_\Delta &  k \alpha_\Delta + \Delta ( 1- \bar\gamma_\Delta)  \\ \bar\Delta \gamma_\Delta + \bar k \bar\alpha_\Delta & \bar k ( \bar\gamma_\Delta - \bar\gamma_0) + \bar\Delta \alpha_\Delta \end{array} \right) 
\end{equation}
where $k$ denotes the left upper entry of $H_\Delta$ (and $H_0$). From
(\ref{def:bcs}) and (\ref{ent:ref})--(\ref{tci2}) we conclude that
\begin{align}\nonumber
  &\F^{\rm BCS}(\Gamma_\Delta) - \F^{\rm BCS}(\Gamma_0) \\ \nonumber &
  = - \frac 1{2\beta} \Tr\left[ \ln(1+e^{-\beta
      H_\Delta})-\ln(1+e^{-\beta H_0})\right] \\ \nonumber
  & \quad - h^{2-2d} \int_{\calC\times \R^d} V(\tfrac {x-y}h)\left|\tfrac 12 (\psi(x)+\psi(y))\alpha_0(\tfrac{x-y}h)\right|^2\, \frac{dx\,dy}{(2\pi)^d} \\
  & \quad + \int_{\calC\times\R^d} V(\tfrac{x-y}h)\left|
    \frac{h^{1-d}}{2(2\pi)^{d/2}}
    \left(\psi(x)+\psi(y)\right)\alpha_0(\tfrac{x-y}h)-
    \alpha_\Delta(x,y)\right|^2\,{dx\,dy} \label{fund:up}
\end{align}
where the $\Tr$ in the first term on the right side has to be
understood as the sum of the traces per unit volume of the diagonal
entries of the $2\times 2$ matrix-valued operator. In general, the
operator $ \ln(1+e^{-\beta H_\Delta})-\ln(1+e^{-\beta H_0})$ is not
trace class if $\Delta$ is not, as can be seen from
(\ref{tci1})--(\ref{tci2}). In the evaluation of $\F^{\rm
  BCS}(\Gamma_\Delta)-\F^{\rm BCS}(\Gamma_0)$ only the diagonal terms
of (\ref{tci2}) enter, however.

The first term on the right side of (\ref{fund:up}) was calculated in
Theorem~\ref{thm:scl} above. Since $\beta = \beta_c + O(h^2)$, we can
replace $\beta$ by $\beta_c$ in all the terms of order $h^4$, yielding
an error of order $h^6$. For the term of order $h^2$, we obtain
\begin{align}\nn
  & {h^2 \beta} \int_{\R^d} t(q)^2 \, g_0(\beta(q^2-\mu))\, dq \\ \nn
  & = {h^2 \beta_c} \int_{\R^d} t(q)^2 \, g_0(\beta_c(q^2-\mu))\, dq
  \\ & \quad + {h^4 \beta_c} D \int_{\R^d} t(q)^2 \left(
    g_0(\beta_c(q^2-\mu)) - \beta_c(q^2-\mu)
    g_1(\beta_c(q^2-\mu))\right) \, dq + O(h^6)\,,
\end{align}
where $D = (T_c - T)/(h^2 T_c)$.  The term in the last line of order
$h^4$ equals
\begin{equation}
  \frac {h^4 \beta_c}{2} D   \int_{\R^d} t(q)^2 \cosh^{-2}(\tfrac 12\beta_c(q^2-\mu))
  \, dq
  =   \frac {h^4 \beta_c^2}{2}   \int_{\R^d} t(q)^4  \,\frac{g_1(\beta_c(q^2-\mu))}{q^2-\mu}\, dq
\end{equation}
according to our normalization (\ref{normal}).

The second term on the right side of (\ref{fund:up}) can be rewritten
as
\begin{align}\nonumber
  & - h^{2(1-d)} \int_{\R^d\times \calC} V(h^{-1}(x-y))\left|\tfrac 12
    (\psi(x)+\psi(y))\alpha_0(h^{-1}(x-y))\right|^2\, dx\,dy \\
  \nonumber & = - h^{2-d} \sum_{p\in (2\pi \Z)^d} |\widehat \psi(p)|^2
  \int_{\R^d} V(x) |\alpha_0(x)|^2 \cos^2(\tfrac h2 p\cdot x)\, dx \\
  & = \frac {h^{2-d} \beta_c }{16} \sum_{p} |\widehat \psi(p)|^2
  \int_{\R^d} t(q) \, g_0(\beta_c(q^2-\mu)) \, \left(2 t(q) + t(q-hp)
    + t(q+hp)\right) \, dq\,. \label{term2}
\end{align}
By writing
\begin{align}\nn
  & 2t(q) + t(q-hp) + t(q+hp) \\ & = 4 t(q) + {h^2} \left[ (p\cdot
    \nabla)^2 t\right]\! (q) + \frac{h^4}{6} \int_{-1}^1 \left[
    (p\cdot \nabla)^4t\right]\!(q+shp) (1-|s|)^3 \, ds
\end{align}
we observe that (\ref{term2}) equals
\begin{align}\nn
  & \frac {h^{2-d} \beta_c }{4} \| \psi\|_2^2 \int_{\R^d} t(q)^2 \, g_0(\beta_c(q^2-\mu))  \, dq\\
  & + \frac {h^{4-d} \beta_c }{16} \sum_{i,j=1}^d \langle \partial_i
  \psi| \partial_j\psi\rangle \int_{\R^d} t(q)
  \left[ \partial_i\partial_j t\right]\! (q) \, g_0(\beta_c(q^2-\mu))
  \, dq + O(h^{6-d})\,,
\end{align}
where the error term is bounded by
\begin{equation}
  C h^{6-d} \|\psi\|_{H^2}^2 \int_{\R^d} |V(x)| |\alpha_0(x)|^2 |x|^4 \, dx\,.
\end{equation}
The latter integral was shown to be finite in Proposition~\ref{prop1}.

If $V\leq 0$, the last term in (\ref{fund:up}) can be dropped for an
upper bound, but we do not need to make this assumption. Since $V$ is
relatively bounded with respect to the Laplacian, we can bound the
term by an appropriate $H^1$ norm. Recall the definition of the $H^1$
norm of a periodic operator in (\ref{def:h1}).  For general periodic
operators $O$, we have the bound
\begin{align}\nn
  & \left| \int_{\calC\times\R^d} V(h^{-1}(x-y))\left| O(x,y)
    \right|^2\,dx\,dy \right| \\ \nn& \leq \left\|
    (1-h^2\nabla^2)^{-1/2} V(h^{-1}(\,\cdot\, -y)) (1-h^2
    \nabla^2)^{-1/2} \right\| \, \|O\|_{H^1}^2 \\ & = \left\|
    (1-\nabla^2)^{-1/2} V(\,\cdot\,) (1- \nabla^2)^{-1/2} \right\| \,
  \|O\|_{H^1}^2\,.
\end{align}
The operator in question can be written as
\begin{equation}
  O  = \alpha_\Delta  - \tfrac h 2 \left ( \psi(x) \widehat\alpha_0(-ih\nabla) + \widehat\alpha_0(-ih\nabla) \psi(x)\right) \,.
\end{equation} 
Recall the definition of $\varphi$ in (\ref{def:varphi}). It equals
$\widehat \alpha_0$ for $\beta = \beta_c$. Hence
\begin{align}\nn
  O & = \alpha_\Delta - \tfrac h 2 \left ( \psi(x) \varphi(-ih\nabla)
    + \varphi(-ih\nabla) \psi(x)\right) \\ & \quad + \tfrac h 4 \left(
    \psi(x) \eta(-ih\nabla) + \eta(-ih\nabla)\psi(x)
  \right)\,, \label{new:term}
\end{align}
with
\begin{equation}
  \eta(q) =  \left( \beta g_0(\beta(q^2 -\mu)) - \beta_c g_0(\beta_c(q^2-\mu)) \right) t(q)\,.
\end{equation}
For the term in the first line of~(\ref{new:term}), we can apply
Theorem~\ref{lem3} to bound its $H^1$ norm by $C h^{3-d/2}
\|\psi\|_{H^2}$. The $H^1$ norm of the term in the second line of
(\ref{new:term}) can be bounded by
\begin{equation} \label{518} C h^{1-d/2} \|\psi\|_{H^1} \left(
    \int_{\R^d} |\eta(q)|^2 (1+q^2) \, dq \right)^{1/2}\,.
\end{equation}
It is easy to see that $|\eta(q)| \leq C(\beta-\beta_c)
|t(q)|/(1+q^2)$, hence (\ref{518}) is bounded by
$Ch^{3-d/2}\|\psi\|_{H^1}$.

For $\psi$, we shall take a minimizer of the GL functional
(\ref{GLfunct}). Under our Assumption~\ref{as1} on $W$ and $A$, it is
easily seen to be in $H^2$. For this choice of $\psi$, we thus have
\begin{equation}
  F^{\rm BCS}(T_c(1-h^2D),\mu) \leq \F^{\rm BCS}(\Gamma_\Delta) - F^{\rm BCS}(\Gamma_0) \leq  h^{4-d} \left( E^{\rm GL} - B_3 + C h \right)
\end{equation}
for small $h$.  This completes the proof of the upper bound.

\section{Proof of Theorem~\ref{thm:main}: Lower Bound, Part
  A}\label{sec:low1}

Our proof of the lower bound on $F^{\rm BCS}(T,\mu)$ in
Theorem~\ref{thm:main} consists of two main parts.  The goal of this
first part is to show the following.  Let again $\Gamma_0$ denote the
normal state defined in (\ref{def:gamma0}), which is the minimizer of
$\F^{\rm BCS}$ in the non-interacting case $V=0$.  We claim that for
any state $\Gamma$ satisfying $\F^{\rm BCS}(\Gamma)\leq \F^{\rm
  BCS}(\Gamma_0)$, we can decompose its off-diagonal part $\alpha$ as
\begin{equation}\label{defpsinew}
  \alpha = \tfrac h 2 \big( \psi(x) \widehat \alpha_0(-ih\nabla) + \widehat \alpha_0(-ih\nabla) \psi(x)\big)  + \xi
\end{equation}
for some periodic function $\psi$ with $H^1(\calC)$ norm bounded
independent of $h$, and with $\|\xi\|_{H^1} \leq O(h^{2-d/2})$, where
we use again the definition (\ref{def:h1}) for the $H^1$ norm of a
periodic operator. This latter bound has to be compared with the $H^1$
norm of the first part of (\ref{defpsinew}), which is $O(h^{1-d/2})$
(for fixed $\psi\neq 0$.)

The remainder of this section contains the proof of
(\ref{defpsinew}). It is divided into three steps.
 
\subsection{Step 1}
 
We claim that for any state $\Gamma$ of the form (\ref{def:gamma})
satisfying $\F^{\rm BCS}(\Gamma)\leq \F^{\rm BCS}(\Gamma_0)$, we have
that
\begin{equation}\label{eq:step1}
  \frac {4T} 5\, \Tr ( \alpha\bar\alpha )^2 + \int_{\calC} \langle \alpha(\,\cdot\,,y) | K_T^{A,W} + V(h^{-1}(\,\cdot\,-y)) | \alpha(\,\cdot\,,y)\rangle \, {dy}\leq 0\,.
\end{equation}
Here, $K_T^{A,W}$ denotes the operator
\begin{equation}\label{def:ktaw}
  K_T^{A,W} =  \frac {\left(-i h \nabla + h A(x) \right)^2 -\mu + h^2 W(x)}{ \tanh\left( \tfrac \beta 2 \left( \left(-i h \nabla + h A(x) \right)^2 -\mu + h^2 W(x)\right) \right)}\,,
\end{equation}
with $\beta = 1/T$.  In (\ref{eq:step1}), it acts on the $x$ variable
of $\alpha(x,y)$, and $\langle\, \cdot \, | \, \cdot \,\rangle$
denotes the standard inner product on $L^2(\R^d)$. Note that
$K_T^{0,0}$ differs from the operator $K_T$ defined in (\ref{def:kt})
by a scaling by $h$. For $T=T_c$, the ground state of $K_{T_c}^{0,0} +
V(h^{-1}(\,\cdot\,-y)$ equals $h^{-d/2} \alpha_0(h^{-1}(x-y))$ (up to
an $h$-independent normalization).

Using that $\left(-i h \nabla + h A(x) \right)^2 -\mu + h^2 W(x) = T
\ln ( (1-\gamma_0)/\gamma_0)$, we may write, for any state $\Gamma$,
\begin{equation}
  \F^{\rm BCS}(\Gamma) - \F^{\rm BCS}(\Gamma_0) = \tfrac 12 T\, \H(\Gamma,\Gamma_0) + \int_{\calC\times\R^d} V(h^{-1}(x-y)) |\alpha(x,y)|^2\, {dx \, dy}\,,
\end{equation}
where $\H$ denotes the relative entropy
\begin{equation}\label{relent}
  \H(\Gamma,\Gamma_0)= \Tr\left[ \Gamma \left( \ln\Gamma - \ln\Gamma_0\right) + 
    (1-\Gamma) \left( \ln(1-\Gamma )- \ln (1-\Gamma_0)\right)\right]\,.
\end{equation}
We have the following lower bound.

\begin{Lemma}\label{lem:klein}
  For any $0\leq \Gamma\leq 1$ and any $\Gamma_0$ of the form
  $\Gamma_0 = (1+e^{H})^{-1}$,
  \begin{equation}\label{eq:lem:klein}
    \H(\Gamma,\Gamma_0) \geq  \Tr\left[ \frac {H}{\tanh (H/2)} \left( \Gamma - \Gamma_0\right)^2\right]  + \frac 43 \Tr\left[ \Gamma(1-\Gamma) - \Gamma_0(1-\Gamma_0)\right]^2\,.
  \end{equation}
\end{Lemma}

A similar bound as (\ref{eq:lem:klein}), without the last positive
term, was used in \cite{HLS}.

\begin{proof}
  It is elementary (but tedious) to show that for real numbers
  $0<x,y<1$,
  \begin{equation}
    x \ln\frac xy + (1-x) \ln \frac{1-x}{1-y} \geq \frac { \ln \frac{1-y}{y} } {1-2y} (x-y)^2 + \frac 43  \left( x(1-x) - y(1-y) \right)^2\,.
  \end{equation}
  The result then follows from Klein's inequality \eqref{klein}.
\end{proof}

Note that for our $\Gamma_0$, $H$ equals $\beta H_0$, which is
diagonal as an operator-valued $2\times 2$ matrix. Hence also
$H_0/\tanh(\beta H_0/2)$ is diagonal. Its diagonal entries are
$\beta\, K_{T}^{A,W}$ and $\beta\, \overline{K}_T^{A,W}$,
respectively, where $K_T^{A,W}$ is given in (\ref{def:ktaw}) above.
Hence
\begin{equation}\label{kleincon}
  \Tr\left[ \frac {H_0}{\tanh \big(\tfrac \beta 2 H_0\big)} \left( \Gamma - \Gamma_0\right)^2 \right] 
  = 2\, \Tr \left[ K_T^{A,W} (\gamma-\gamma_0)^2\right] + 2 \, \Tr \left[ K_T^{A,W}  \alpha \bar \alpha\right] \,.
\end{equation}
Since $x/\tanh(x/2) \geq 2$, we can replace $K_T^{A,W}$ by $2 T$ for a
lower bound. We shall use this in the first term on the right side of
(\ref{kleincon}).

For the last term in (\ref{eq:lem:klein}), we use
\begin{equation}
  \Tr\left[ \Gamma(1-\Gamma) - \Gamma_0(1-\Gamma_0)\right]^2 \geq 2\, \Tr \left [ \gamma(1-\gamma)- \gamma_0(1-\gamma_0) - \alpha\bar\alpha\right]^2\,.
\end{equation}
We claim that
\begin{equation}
  2  \, \Tr (\gamma -\gamma_0)^2 + \frac {4 }{3}\, \Tr \left [ \gamma(1-\gamma)- \gamma_0(1-\gamma_0) - \alpha\bar\alpha\right]^2 \geq \frac 4 5\, \Tr ( \alpha\bar\alpha )^2\,.
\end{equation}
This follows easily from the triangle inequality
\begin{equation}
  \|\alpha\bar\alpha\|_2 \leq 
  \|  \gamma(1-\gamma)- \gamma_0(1-\gamma_0) - \alpha\bar\alpha \|_2 + \|\gamma(1-\gamma)-\gamma_0(1-\gamma_0)\|_2
\end{equation}
together with the fact that
\begin{equation}
  \|\gamma(1-\gamma)-\gamma_0(1-\gamma_0)\|_2 \leq \|\gamma-\gamma_0\|_2\,,
\end{equation}
which can be seen using Klein's inequality (\ref{klein}), for
instance.

This completes the proof of (\ref{eq:step1}).

\subsection{Step 2} Recall that $K_{T_c}^{0,0} + V(h^{-1}(\,\cdot\,))$
is non-negative and has a non-degenerate isolated eigenvalue
zero. Hence it will be convenient to replace $K_{T}^{A,W}$ by
$K_{T_c}^{0,0}$ in (\ref{eq:step1}). The following lemma quantifies
the effect of such a replacement.

\begin{Lemma} \label{lem:mon} For $T = T_c - O(h^2)$ and $h$ small
  enough,
  \begin{equation}
    K_T^{A,W} + V(h^{-1}(\,\cdot\,-y)) \geq \frac 18 \left( K_{T_c}^{0,0} + V(h^{-1}(\,\cdot\, -y) ) \right) - C h^2 
  \end{equation}
  for a constant $C>0$ depending only on $\|W\|_\infty$, $\|A\|_{C^1}$
  and $h^{-2}(T-T_c)$.
\end{Lemma}

The proof shows that the prefactor $1/8$ can be replaced by any number
less than one, at the expense of an increase in the constant $C$.

\begin{proof}
  We start by noting that \cite[(4.3.91)]{AS}
  \begin{equation}\label{expa}
    \frac{ x}{\tanh(x/2)} = 2 +  \frac 12 \sum_{k=1}^\infty  \frac{x^2}{x^2/4 + k^2 \pi^2}=
    2 + \sum_{k=1}^\infty \left( 2- \frac{2 k^2 \pi^2}{x^2/4 + k^2 \pi^2}\right)\,.
  \end{equation}
  In particular, this function is operator monotone as a function of
  $x^2$. Let $\widetilde p = -ih\nabla + h A(y)$ and $\widetilde A(x)
  = A(x) - A(y)$. Since the ground state of the operator
  $K_{T_c}^{0,0}+ V(h^{-1}(\,\cdot\,-y))$ is localized within a
  distance $O(h)$ of $y$, $\widetilde A(x)$ is, as far as the ground
  state is concerned, effectively a perturbation of order $h$ since
  $A$ is assumed to be Lipschitz continuous.  Using Schwarz's
  inequality,
  \begin{align}\nn
    & \left[ \left( \widetilde p + h \widetilde A(x) \right)^2 -\mu +
      h^2 W(x)\right]^2 \\ \nn & \geq (1-\epsilon) \left[ \widetilde
      p^2 + h\, \widetilde p\cdot \widetilde A(x) + h\, \widetilde
      A(x)\cdot \widetilde p -\mu \right]^2 - \frac {h^4}{\epsilon}
    \|W + \widetilde A^2\|_\infty^2 \\ & \geq (1-\epsilon)^2 \left[
      \widetilde p^2-\mu\right]^2 - \frac {h^4}{\epsilon} \|W +
    \widetilde A^2\|_\infty^2 - \frac{h^2}{\epsilon} \left[ \widetilde
      p\cdot \widetilde A(x) + \widetilde A(x)\cdot \widetilde
      p\right]^2\,. \label{sss}
  \end{align}
  We write the right side as $\beta^{-2} (R-Q)$, where
  \begin{equation}
    Q = \frac{h^2\beta^2}\epsilon \left[  \widetilde p\cdot \widetilde A(x) +  \widetilde A(x)\cdot \widetilde p\right]^2
  \end{equation}
  and
  \begin{equation}
    R = \beta^2  (1-\epsilon)^2 \left[ \widetilde p^2-\mu\right]^2 - \frac {h^4 \beta^2}{\epsilon} \|W + \widetilde A^2\|_\infty^2 \,.
  \end{equation}

  As long as $R-Q + 4\pi^2>0$, we can use the operator monotonicity of
  (\ref{expa}) to obtain a lower bound.  This condition is certainly
  satisfied for small enough values of $\epsilon$, $h^2/\epsilon$ and
  $h^4/\epsilon$. Hence we can use the resolvent identity to conclude
  that
  \begin{align}\nn
    \beta\, K_T^{A,W} & \geq 2 + \frac 12
    \sum_{k=1}^\infty\frac{R}{k^2\pi^2 + R/4} \\ \nn & \quad -\frac 12
    \sum_{k=1}^\infty \frac{k^2 \pi^2}{k^2 \pi^2 + R/4}Q \frac{1}{k^2
      \pi^2 + R/4} \\ & \quad - \frac 18\sum_{k=1}^\infty \frac{k^2
      \pi^2}{k^2 \pi^2 + R/4} Q \frac{1}{k^2 \pi^2 + (R - Q)/4} Q
    \frac{1}{k^2 \pi^2 + R/4} \,.
  \end{align}
  Since also $R- Q > Q -4\pi^2$ for $\epsilon$, $h^2/\epsilon$ and
  $h^4/\epsilon$ small enough, we can bound
  \begin{equation}
    Q  \frac{1}{k^2 \pi^2 + (R - Q)/4}   Q \leq  4 Q 
  \end{equation}
  for all $k\geq 1$. We thus obtain the lower bound
  \begin{equation}
    K_T^{A,W} \geq \frac 1 \beta\left( 2 + \frac 12 \sum_{k=1}^\infty\frac{R}{k^2\pi^2 + R/4} \right) -   \frac 1 \beta \sum_{k=1}^\infty \frac{k^2 \pi^2}{k^2 \pi^2 + R/4}Q \frac{1}{k^2 \pi^2 + R/4}  \,. \label{3t} 
  \end{equation}

  We start with deriving a lower bound on the first term on the right
  side of (\ref{3t}), and defer the discussion of the second term to
  (\ref{bound:E}) {\it et seq}.  For $h^4/\epsilon$ small enough, the
  first term on the right side of (\ref{3t}) is bounded from below by
  \begin{equation}
    (1-\epsilon)^2\beta K_T^{A_y,0} -  \frac{h^4}{\epsilon}\beta^2\|W+\widetilde A^2\|_\infty^2 \frac{\pi^2}{12}\frac 1{\pi^2 - h^4 \epsilon^{-1}\beta^2 \|W+\widetilde A^2\|_\infty^2} \,,
  \end{equation}
  where we denote by $A_y$ the constant vector potential
  $A(y)$. Moreover, it is elementary to show that
  \begin{equation}
    K_T^{A_y,0} \geq K_{T_c}^{A_y,0} - 2 (T_c-T)
  \end{equation}
  for $T\leq T_c$. We further have
  \begin{align}\nn
    (1-2 \epsilon) K_{T_c}^{A_y,0} + V(h^{-1}(\,\cdot\,-y)) & = (1 - 3
    \epsilon ) \left( K_{T_c}^{A_y,0} + V(h^{-1}(\,\cdot\,-y)) \right)
    \\ \nn & \quad + \epsilon \left( K_{T_c}^{A_y,0} + 3
      V(h^{-1}(\,\cdot\,-y)) \right) \\ & \geq (1 - 3 \epsilon )
    \left( K_{T_c}^{A_y,0} + V(h^{-1}(\,\cdot\,-y)) \right) - C
    \epsilon\,, \label{kay}
  \end{align}
  where $-C$ is the ground state energy of $K_{T_c}^{A_y,0}+ 3V$,
  which is bounded by our assumptions on $V$. It is also independent
  of $A_y$ since $A_y$ may be replaced by zero by a unitary (gauge)
  transformation.

  We now want to get rid of the constant vector potential $A(y)$ in
  $K_{T_c}^{A_y,0}$ on the right side of (\ref{kay}).  We claim that
  \begin{equation}
    K_{T_c}^{A_y,0} + V(h^{-1}(\,\cdot\, -y) ) \geq  \frac 12 \left( K_{T_c}^{0,0} + V(h^{-1}(\,\cdot\, -y) ) \right) - C h^2 A(y)^2 \,.
  \end{equation}
  This follows from the fact that
  \begin{equation}
    K_{T_c}^{A_y,0} \geq  \frac 12 K_{T_c}^{0,0}  + \frac 12 K_{T_c}^{2A_y,0} - C h^2 A(y)^2
  \end{equation}
  (since the function $p \mapsto (p^2-\mu)/\tanh[\frac \beta 2
  (p^2-\mu)]$ has a bounded Hessian) and
  \begin{equation}
    K_{T_c}^{2 A_y,0} + V(h^{-1}(\,\cdot\, -y) )  \geq 0
  \end{equation}
  by the definition of $T_c$.

  For $3\epsilon \leq \tfrac 12$ and $T_c - T = O(h^2)$, we conclude
  that
  \begin{align}\nn
    K_T^{A,W} + V(h^{-1}(\,\cdot\, -y) ) & \geq \frac 14 \left(
      K_{T_c}^{0,0} + V(h^{-1}(\,\cdot\, -y) )\right) \\ & \quad - C
    \left(h^2+\epsilon + h^4 \epsilon^{-1}\right) - E
  \end{align}
  where $E$ denotes the last term in (\ref{3t}). Let $P$ denote the
  projection onto the ground state of
  $K_{T_c}^{0,0}+V(h^{-1}(\,\cdot\, -y) )$, given by the function
  $\alpha_0(h^{-1}(\,\cdot\,-y))$. Let also $P^c = 1- P$. Since $E$ is
  positive, it follows from Schwarz's inequality that
  \begin{equation}\label{bound:E}
    E \leq 2 P E P + 2 P^c E P^c \,.
  \end{equation}
  From the assumption on $A\in C^1$ it follows easily that $0\leq Q
  \leq C h^2 \epsilon^{-1} (1+\widetilde p^2)$. This immediately
  implies that $E\leq C h^2 \epsilon^{-1} (1+\widetilde p^2)\leq C
  h^2\epsilon^{-1} (1-h^2\nabla^2)$ (since $A$ is bounded).  We shall
  choose $\epsilon = O(h^2)$, with $h^2/\epsilon$ small enough to
  ensure that
  \begin{equation}
    \frac 18 \left( K_{T_c}^{0,0} + V(h^{-1}(\,\cdot\, -y) )\right) - 2 P^c E P^c  \geq 0\,. 
  \end{equation}
  That this can be done follows from the above bound on $E$ and the
  fact that $ K_{T_c}^{0,0}+ V(h^{-1}(\,\cdot\, -y) \geq \nu
  P^c(1-h^2\nabla^2)P^c$ for some $\nu>0$.

  It remains to show that $P E P \leq O(h^2)$. Since $P$ has rank one,
  we can write
  \begin{equation}
    P EP = C h^{-d} \left\langle \alpha_0( h^{-1} ( \, \cdot\, - y) \left| E \right|\alpha_0(h^{-1}(\, \cdot\, - y) \right\rangle \, P 
  \end{equation}
  for some $h$-independent constant $C$ determined by the
  normalization of $\alpha_0$.  We have
  \begin{align}\nonumber
    Q & = \frac{h^2\beta^2}{\epsilon} \left( 2 \widetilde p\cdot
      \widetilde A(x) + ih \, {\rm div\,} A(x) \right) \left( 2
      \widetilde A(x) \cdot \widetilde p - i h\, {\rm div\, } A(x)
    \right) \\ & \leq \frac{h^2\beta^2}{\epsilon} \left( 8 \widetilde
      p\cdot \widetilde A(x) \, \widetilde A(x) \cdot \widetilde p + 2
      h^2 \left( {\rm div\,} A(x)\right)^2 \right)\,. \label{boundq}
  \end{align}
  Note that the last term is bounded by our assumption $A\in C^1$.
  From the Lipschitz continuity of $A$ it follows that $|\widetilde
  A(x)| \leq C |x-y|$, and hence
  \begin{equation}
    \widetilde p\cdot \widetilde A(x) \, \widetilde A(x) \cdot \widetilde p \leq C^2  \widetilde p \cdot |x-y|^2 \, \widetilde p\,.
  \end{equation}
  We thus have $P E P \leq C h^2 P \sum_{k} k^2 \lambda_k$ where
  \begin{align}\nn
    \lambda_k & = \frac{h^2\beta}{\epsilon} \sum_{i=1}^d \int_{\R^d}
    \left| \nabla_p \frac{p_i + h A_i(y)}{k^2\pi^2 + \tfrac 14 \beta^2
        (1-\epsilon)^2 \left[ (p + h A_y)^2 -\mu\right]^2 - \delta}
      \widehat \alpha_0(p)\right|^2 \, dp \\ & \quad +
    \frac{h^2\beta}{\epsilon}\left\| {\rm div\,} A \right\|_\infty^2
    \frac{ 1}{\left( k^2\pi^2 -\delta\right)^2} \,,
  \end{align}
  with $\delta = O(h^2)$ denoting $\delta = \tfrac 14 h^4 \beta
  \epsilon^{-1} \|W+\widetilde A\|_\infty^2$. For small enough $h$, we
  can thus bound
  \begin{equation}
    \lambda_k \leq C \frac{h^2}{\epsilon k^4} \int_{\R^d} \left(1+|x|^2\right) \left( \left|\nabla \alpha_0(x)\right|^2 + \left|\alpha_0(x)\right|^2 \right) dx \,.
  \end{equation}
  The latter integral is finite, as proved in
  Proposition~\ref{prop:reg} in Section~\ref{sec:alpha0}. Hence
  $\sum_k k^2 \lambda_k$ is bounded, uniformly in $h$ for small
  $h$. This completes the proof.
\end{proof}
 
\subsection{Step 3} In combination with Lemma~\ref{lem:mon}, we
conclude from (\ref{eq:step1}) that for $T= T_c - O(h^2)$,
\begin{equation}\label{eq1}
  \frac {4T} 5\, \Tr [ \alpha\bar\alpha ]^2 + \frac 18 \int_{\calC} \langle \alpha(\,\cdot\,,y) | K_{T_c}^{0,0} + V(h^{-1}(\,\cdot\,-y)) | \alpha(\,\cdot\,,y)\rangle \,  {dy}\leq C h^2 \|\alpha\|_2^2 
\end{equation}
for any state $\Gamma$ with $\F^{\rm BCS}(\Gamma)\leq \F^{\rm
  BCS}(\Gamma_0)$. We shall now show that this inequality implies
(\ref{defpsinew}).

Recall that the operator $ K_{T_c}^{0,0} + V(h^{-1}(\,\cdot\,-y))$ on
$L^2(\R^d)$ has a unique ground state, proportional to
$\alpha_0(h^{-1}(x-y))$, with ground state energy zero, and a gap
above. Normalize $\alpha_0$ as in (\ref{normal}), and let
\begin{equation}
  \psi(y) =  (2\pi)^{d/2} \left(h \int_{\R^d} |\alpha_0(x)|^2 \, dx  \right)^{-1} \int_{\R^d} \alpha_0(h^{-1}(x-y)) \alpha(x,y) \,dx\,.
\end{equation}
Note that $\psi$ is a periodic function.  If we write
\begin{equation}\label{dec1}
  \alpha(x,y)=\psi(y) \frac{ h^{1-d}}{(2\pi)^{d/2}} \alpha_0(h^{-1}(x-y))  + \xi_0(x,y)
\end{equation}
the gap in the spectrum of $K_{T_c}^{0,0} + V(h^{-1}(\,\cdot\,-y))$
above zero, together with (\ref{eq1}), yields $\|\xi_0\|_2 \leq
O(h)\|\alpha\|_2$. We can also symmetrize and write
\begin{equation}
  \alpha(x,y) = \tfrac 12 \left( \psi( x)+\psi(y)\right) \frac{h^{1-d}}{(2\pi)^{d/2}} \alpha_0(h^{-1}(x-y)) + \xi(x,y) \,,
\end{equation}
again with $\|\xi\|_2\leq O(h)\|\alpha\|_2$. In order to complete the
proof of (\ref{defpsinew}), we need to show that $\|\psi\|_{H^1}$ is
bounded independently of $h$, and that the $H^1$ norm of $\xi$ is
bounded by $O(h^{2-d/2})$.

An application of Schwarz's inequality yields
\begin{equation}
  \int_\calC | \psi(x)|^2\, dx \leq  (2\pi)^d h^{d-2} \frac{ \|\alpha\|_2^2}{ \int_{\R^d} |\alpha_0(x)|^2 dx } \,.
\end{equation}
Moreover,
\begin{equation}
  \|\alpha\|_2^2 \leq \frac{h^{2-d}}{(2\pi)^d} \int_\calC |\psi(x)|^2 dx \, \int_{\R^d} |\alpha_0(x)|^2 dx + \|\xi_0\|_2^2\,.
\end{equation}
Since $\|\xi_0\|_2\leq O(h) \|\alpha\|_2$, we see that
\begin{equation}\label{schw1}
  \|\alpha\|_2^2 \leq (1 + O(h^2)) \frac{ h^{2-d}}{(2\pi)^d} \int_\calC |\psi(x)|^2 dx\, \int_{\R^d} |\alpha_0(x)|^2 dx \,.
\end{equation}

Again by using Schwarz's inequality,
\begin{equation}\label{schw2}
  \int_\calC |\nabla \psi(x)|^2\, dx \leq  (2\pi)^d h^{d-2} \frac { \int_{\R^d\times \calC}  \left| \left(\nabla_x + \nabla_y\right) \alpha(x,y)\right|^2 \, dx\,dy }{ \int |\alpha_0(x)|^2dx } \,.
\end{equation}
In order to bound the latter expression, we use the following lemma.

\begin{Lemma}\label{lem:P} For some constant $C>0$,
  \begin{align}\nonumber
    & h^2 \int_{\R^d\times \calC} \left| \left(\nabla_x +
        \nabla_y\right) \alpha(x,y)\right|^2 \, dx\,dy \\ & \leq C
    \int_{\calC} \langle \alpha(\,\cdot\,,y) | K_{T_c}^{0,0} +
    V(h^{-1}(\,\cdot\,-y)) | \alpha(\,\cdot\,,y)\rangle \,
    dy \label{lem:eq:com}
  \end{align}
  for all periodic and symmetric $\alpha$ (i.e.,
  $\alpha(x,y)=\alpha(y,x)$).
\end{Lemma}

\begin{proof}
  By expanding $\alpha(x,y)$ in a Fourier series
  \begin{equation}
    \alpha(x,y) = \sum_{ p \in (2\pi \Z)^d} e^{i p\cdot (x+y)/2} \widetilde \alpha_p(x-y)
  \end{equation}
  and using that $\widetilde \alpha_p(x) = \widetilde \alpha_p(-x)$
  for all $p \in (2\pi \Z)^d$ we see that (\ref{lem:eq:com}) is
  equivalent to
  \begin{equation}
    K_{T_c}^{\frac 1 2 p,0}  + K_{T_c}^{-\frac 1 2 p,0} +2 \, V(x/h) \geq \frac 2 C h^2 p^2
  \end{equation}
  for all $p\in (2\pi \Z)^d$. This inequality holds for all $p\in
  \R^d$, in fact, for an appropriate choice of $C>0$, as we shall now
  show.

  Since $K_{T_c}^{\frac 12 p,0} \geq C (1+ h^2(-i\nabla + p/2)^2)$, it
  suffices to consider the case of $h p$ small. Let $\kappa$ denote
  the gap in the spectrum of $K^{0,0}_{T_c} + V(h^{-1}\,\cdot\,)$
  above zero, and $\varphi_0^h(x):=h^{-d/2}\varphi_0(x/h)$ its
  normalized ground state. (Note that $\varphi_0$ is normalized to
  one, and hence equals a constant times $\alpha_0$.) Then
  \begin{align}\nn
    & K_{T_c}^{\frac 12 p,0} + K_{T_c}^{-\frac 12 p,0} +2 \, V(x/h) \\
    \nn & \geq \kappa\left[ e^{i x \cdot p/2}\left( 1 -
        |\varphi_0^h\rangle\langle\varphi_0^h|\right) e^{-ix \cdot
        p/2}+ e^{-i x\cdot p/2}\left( 1 -
        |\varphi_0^h\rangle\langle\varphi_0^h|\right) e^{ix\cdot p/2}
    \right] \\ & \geq \kappa\left[ 1 - \left| \int_{\R^d}
        |\varphi_0(x)|^2 e^{-ih x\cdot p} dx\right| \right]\,,
  \end{align}
  where the last expression is simply the lowest eigenvalue of the
  operator on the previous line. Since $\varphi_0$ is reflection
  symmetric,
  \begin{equation}
    \int_{\R^d} |\varphi_0(x)|^2 e^{-ih x\cdot p} dx =  \int_{\R^d} |\varphi_0(x)|^2 \cos\left(h x\cdot p\right) dx = 1 - O(h^2 p^2)
  \end{equation}
  for small $hp$, since $\int_{\R^d} |x|^2 |\varphi_0(x)|^2 dx$ is
  finite by Proposition~\ref{prop:reg}.
\end{proof}

By combining (\ref{lem:eq:com}) with (\ref{schw1}), (\ref{schw2}) and
(\ref{eq1}) we see that $\|\nabla\psi\|_2$ is bounded by a constant
times $\|\psi\|_2$. To conclude the uniform upper bound on the $H^1$
norm of $\psi$, it thus suffices to give a bound on the $L^2$ norm. To
do this, we have to utilize the first term on the left side of
Eq.~(\ref{eq1}).

Recall that $\alpha$ can be decomposed as $\alpha = h \alpha_0 \psi +
\xi_0$, as in (\ref{dec1}), where $\alpha_0$ is short for the operator
$\widehat\alpha_0(-ih \nabla)$. The following lemma gives a lower
bound on $ ( \Tr (\alpha\bar\alpha)^2 )^{1/4}$, the $4$-norm of
$\alpha$.

\begin{Lemma}\label{lem:a4}
  For some $0<C<\infty$ we have
  \begin{align}\nn
    \|\alpha\|_4 & \geq \left[ h^{4-d} \int_\calC |\psi(x)|^4 \, dx
      \int_{\R_d} \widehat\alpha_0(q)^4 \, \frac{dq}{(2\pi)^d} - C
      h^{5-d} \|\psi\|_{H^1(\calC)}^4 \right]_+^{1/4} \\ & \quad - C
    h^{1-d/4} \|\psi\|_2^{1/2} \left ( 1 + C h^{1-d/4} \|\psi\|_4
    \right)^{1/2}\,,\label{eq:lem:a4}
  \end{align}
  where $[\,\cdot\,]_+ = \max\{ 0, \,\cdot\,\}$ denotes the positive
  part.
\end{Lemma}

\begin{proof}
  By the triangle inequality,
  \begin{equation}\label{ft1}
    \|\alpha\|_4 \geq  h \|\alpha_0 \psi\|_4 - \|\xi_0\|_4 \,.
  \end{equation}
  We can bound the last term as $\|\xi_0\|_4^2 \leq \|\xi_0\|_\infty
  \|\xi_0\|_2$. Recall that we have already shown that $\|\xi_0\|_2
  \leq O(h)\|\alpha\|_2 \leq O(h^{2-d/2}) \|\psi\|_2$.  Note that we
  cannot bound the $\infty$-norm simply by the $2$-norm, since the
  norms $\|\,\cdot\,\|_p$ defined via the trace per unit volume in
  (\ref{def:pnorm}) are, in general, not monotone decreasing in $p$.
  We can, however, use that $\|\xi_0\|_\infty\leq
  \|\alpha\|_\infty+h\|\alpha_0\psi\|_\infty$ and that
  $\|\alpha\|_\infty\leq 1$ for any admissible state. Moreover, we
  claim that
  \begin{equation}\label{neqeq}
    \|\alpha_0\psi\|_\infty \leq C_\nu \|\psi\|_4 h^{-d/4}   \left( \int_{\R^d} \left|\alpha_0(x) (1+|x|)^\nu \right|^{4/3}  dx \right)^{3/4}
  \end{equation}
  for $\nu > d$; that the latter integral
  is finite follows from the last statement in
  Proposition~\ref{prop:reg}, which says that $\alpha_0(\,\cdot\,)
  (1+|\,\cdot\,|)^\nu$ is in $L^2(\R^d)$ for {\em any} $\nu>0$.

  Eq.~(\ref{neqeq}) can be obtained with the aid of Young's inequality
  \cite[Thm.~4.2]{LL}, as we now explain. With $\chi$ denoting the
  characteristic function of the unit cube, and $\chi_j(x) =
  \chi(x-j)$ for $j\in \Z^d$, we have, for any $f,g\in L^2(\R^d)$,
  \begin{align}\nn
    & \left| \left\langle f \left| \alpha_0 \psi \right| g \right
      \rangle \right| \\ \nn &
    \leq h^{-d} \sum_{j,k \in \Z^d}  \left( 1+ \tfrac{[|j-k|-\sqrt{d}]_+}h \right)^{-\nu} \\
    \nn & \quad \times \int_{\R^d} \left| \chi_j(x) f(x)
      \alpha_0(\tfrac{x-y}h) \left( 1+ \tfrac{|x-y|}h \right)^\nu
      \psi(y) \chi_k(y) g(y)\right| \,dx\,dy \\ & \leq C h^{-d/4}
    \sum_{j,k \in \Z^d} \left( 1+ \tfrac{[|j-k|-\sqrt{d}]_+}h
    \right)^{-\nu} \|\chi_j f\|_2 \|\chi_k g\|_2 \|\psi\|_{L^4(\calC)}
    \|\alpha_0 (1+|\,\cdot\,|)^\nu\|_{4/3} \,,
  \end{align}
  where we used Young's inequality in the second step.  We can further
  bound $ \|\chi_j f\|_2 \|\chi_k g\|_2 \leq \lambda \|\chi_j f\|_2^2
  + \lambda^{-1} \|\chi_k g\|_2^2$ for $\lambda>0$. After doing the
  sum over $j$ and $k$ and optimizing over $\lambda$, we thus have
  \begin{align}\nn
    & \left| \left\langle f \left| \alpha_0 \psi \right| g \right
      \rangle \right| \\ & \leq C h^{-d/4} \|f\|_2 \|g\|_2 \sum_{j\in
      \Z^d} \left( 1+ \tfrac{[|j|-\sqrt{d}]_+}h \right)^{-\nu}
    \|\psi\|_{L^4(\calC)} \|\alpha_0 (1+|\,\cdot\,|)^\nu\|_{4/3}\,,
  \end{align}
  which yields (\ref{neqeq}).

  We have thus shown that
  \begin{equation}
    \|\xi_0\|_4^2 \leq C h^{2-d/2} \|\psi\|_2 \left ( 1 +  C h^{1-d/4}  \|\psi\|_4 \right)\,.
  \end{equation}
  It remains to investigate the first term in (\ref{ft1}).  A short
  calculation shows that
  \begin{equation}
    \|\alpha_0\psi\|_4^4 = h^{-d} \sum_{p_1,p_2,p_3} \widehat \psi(p_1) \widehat \psi^*(p_2) \widehat\psi(p_3)\widehat\psi^*(-p_1-p_2-p_3) F(h p_1,h p_2,h p_3)\,,
  \end{equation}
  where
  \begin{equation}
    F(p_1,p_2,p_3) =  \int_{\R^3}  \widehat \alpha_0(q) \widehat \alpha_0(q+p_1) \widehat \alpha_0(q+p_1+p_2) \widehat \alpha_0(q+p_1+p_2+p_3) \, \frac{dq}{(2\pi)^d} \,.
  \end{equation}
  Note that $F$ is bounded.  In fact,
  \begin{equation}
    \left| F(p_1,p_2,p_3)\right| \leq (2\pi)^{-d}\| \widehat \alpha_0\|_4^4 \,.
  \end{equation}
  To see that this is finite, one can use that $|\widehat \alpha_0(p)|
  \leq C |t(p)|/(1+p^2)$ (from the definition (\ref{deft})) and that
  $t\in L^q$ with $q$ as in Proposition~\ref{prop:reg}.  In the same
  way, one can show that $F$ has a bounded derivative, and hence
  \begin{equation}
    \left| F(p_1,p_2,p_3) - F(0,0,0)\right| \leq C ( |p_1|+|p_2|+|p_3| ) \,.
  \end{equation}

  We are left with giving a bound on
  \begin{align}\nn
    & \sum_{p_1,p_2,p_3} \left| \widehat \psi(p_1) \widehat
      \psi^*(p_2) \widehat\psi(p_3)\widehat\psi^*(-p_1-p_2-p_3)
    \right| |p_1| \\ \nn & \leq \left( \sum_{p} |\widehat \psi(p)|^2
      |p|^2 \right)^2 \left( \sum_{p} \left| \sum_{p_2,p_3} \left|
          \widehat \psi^*(p_2) \widehat \psi(p_3) \widehat
          \psi^*(-p-p_2-p_3)\right| \right|^2 \right)^{1/2} \\ & =
    \left( \int_{\calC} |\nabla\psi(x)|^2 \, dx \right)^{1/2}
    \left(\int_\calC |\widetilde\psi(x)|^6 dx
    \right)^{1/2}\,, \label{649}
  \end{align}
  where $\widetilde\psi$ is the function whose Fourier transform
  equals $|\widehat\psi|$.  For $d\leq 3$, we can use the Sobolev
  inequality
  \begin{equation}\label{650}
    \left( \int_\calC |\widetilde\psi(x)|^6 dx\right)^{1/3} \leq  C \|\widetilde \psi\|_{H^1(\calC)}^2 = C \|\psi\|_{H^1(\calC)}^2 \,.
  \end{equation}
  Hence
  \begin{equation}
    \|\alpha_0\psi\|_4^4 \geq h^{-d} \int_\calC |\psi(x)|^4 \, dx \, \int_{\R^d} \widehat \alpha_0(q)^4 \, \frac{dq}{(2\pi)^d}  - C h^{1-d} \|\psi\|_{H^1(\calC)}^4\,,
  \end{equation}
  and this completes the proof.
\end{proof}

We have already shown that $\|\nabla\psi\|_2\leq C \|\psi\|_2$, which
also implies that $\|\psi\|_4 \leq C \|\psi\|_2$ via Sobolev's
inequality for functions on the torus. If we use also that $\|\psi\|_4
\geq \|\psi\|_2$ (since the norms of $\psi$ are defined via
integration over the unit cube $\calC$), we conclude from
(\ref{eq:lem:a4}) that $\|\alpha\|_4 \geq C h^{1-d/4} (\|\psi\|_2 - C
\|\psi\|_2^{1/2})$ for $h$ small enough. In combination with
(\ref{eq1}) and (\ref{schw1}) this implies that $\|\psi\|_2 \leq
C$. This shows that the $H^1$ norm of $\psi$ is indeed uniformly
bounded.

It follows that $\|\xi\|_2 \leq O(h^{2-d/2})$. To conclude the proof
of (\ref{defpsinew}), we need to show that also $\|\xi\|_{H^1} \leq
O(h^{2-d/2})$, i.e., $\|\nabla\xi\|_{2}\leq O(h^{1-d/2})$. We have
\begin{equation}
  \xi(x,y) = \xi_0(x,y) + \tfrac 12 \left(\psi(x) - \psi(y)\right)\frac{ h^{1-d}}{(2\pi)^{d/2}} \alpha_0(h^{-1}(x-y)) \,.
\end{equation}
From the definition (\ref{dec1}) it follows easily that $\|
\xi_0\|_{H^1}\leq O(h^{2-d/2})$, since
\begin{align}\nn
  \|\xi_0\|_{H^1}^2 & \leq C \int_\calC \left\langle
    \xi_0(\,\cdot\,,y) \left| K_T^{0,0} + V(h^{-1}(\,\cdot\, -y)) + 1
    \right| \xi_0(\,\cdot\,,y)\right\rangle dy \\ &\leq C \left( h^2
    \|\alpha\|_2^2 + \|\xi_0\|_2^2\right) \leq O(h^{4-d})
\end{align}
using (\ref{eq1}).  Recall that the definition of the $H^1$ norm in
(\ref{def:h1}) is not symmetric, hence this does not immediately imply
a bound on the $H^1$ norm of $\xi$. However, we can estimate
\begin{equation}
  h^{2(1-d)} \int_\calC |\nabla\psi(x)|^2 \int_{\R^d} |\alpha_0(h^{-1}(x-y))|^2 dx\, dy \leq O(h^{2-d})\,.
\end{equation}
Finally
\begin{align}\nn
  &h^{-2d} \int_{\calC\times\R^d} |\psi(x)-\psi(y)|^2 |\nabla
  \alpha_0(h^{-1}(x-y))|^2\,dx \, dy \\ & = 4 h^{-d} \sum_{p\in
    (2\pi\Z)^d} |\widehat \psi(p)|^2 \int_{\R^d} |\nabla
  \alpha_0(x)|^2 \sin^2\left( \tfrac 12 h p\cdot x \right) \, dx \leq
  O(h^{2-d}) \,,
\end{align}
since the $H^1$ norm of $\psi$ is bounded and $\int |\nabla\alpha_0|^2
|x|^2dx$ is finite, as shown in Proposition~\ref{prop:reg} in
Section~\ref{sec:alpha0}.  This proves the claim.

\section{Proof of Theorem~\ref{thm:main}: Lower Bound, Part
  B}\label{sec:low2}

We now conclude the proof of the lower bound of
Theorem~\ref{thm:main}. For convenience, we shall divide the proof
into 2 steps.

\subsection{Step 1}

Let $\Gamma$ be a state with $\F^{\rm BCS}(\Gamma) \leq \F^{\rm
  BCS}(\Gamma_0)$. In the previous section, we have shown that the
off-diagonal part $\alpha$ of $\Gamma$ can be decomposed in the form
(\ref{defpsinew}).  Given $\psi$ defined in (\ref{defpsinew}) and some
$\epsilon>0$, define $\psi_<$ via its Fourier transform
\begin{equation}
  \widehat
  \psi_<(p) = \widehat \psi(p) \theta(\epsilon h^{-1} - |p|) \,, 
\end{equation}
where $\theta(t) = 1$ for $t\geq 1$, and $0$ otherwise.  The function
$\psi_<$ is thus smooth, and $\|\psi_<\|_{H^2}\leq C (1 + \epsilon
h^{-1})$ since $\psi$ is bounded in $H^1$. We shall choose
$h<\epsilon<1$.

Let also $\psi_> = \psi - \psi_<$. Since $\psi$ is bounded in $H^1$,
the $L^2(\calC)$ norm of $\psi_>$ is bounded by $O(h \epsilon^{-1})$.
We absorb the part $\frac
12(\psi_>(x)+\psi_>(y))\alpha_0(h^{-1}(x-y))$ into $\xi$, and write
\begin{equation}\label{defpsi}
  \alpha(x,y) = \tfrac 12 \left( \psi_<(x)+\psi_<(y)\right) \frac{h^{1-d}}{(2\pi)^{d/2}} \alpha_0(h^{-1}(x-y)) +\sigma(x,y)
\end{equation}
where
\begin{equation}\label{def:sigma}
  \sigma (x,y) = \xi(x,y) + \tfrac 12 \left( \psi_>(x)+\psi_>(y)\right) \frac{h^{1-d}}{(2\pi)^{d/2}} \alpha_0(h^{-1}(x-y)) \,.
\end{equation}
In the previous section, we have shown that $\|\xi\|_{H^1}\leq
O(h^{2-d/2})$.  From the bound $\|\psi_>\|_2 \leq O(h\epsilon^{-1})$
it thus follows that $\|\sigma\|_2 \leq
O(h^{2-d/2}\epsilon^{-1})$. Note that we cannot conclude the same
bound for the $H^1$ norm of $\sigma$, however.

As in (\ref{def:delta}), let $\Delta = -\frac 12 (\psi_<(x)
t(-ih\nabla) + t(-ih\nabla) \psi_<(x))$. Its integral kernel is given
in (\ref{int:delta}), with $\psi$ replaced by $\psi_<$.  Let
$H_\Delta$ be the corresponding Hamiltonian defined in
(\ref{hdelta}). We can write
\begin{align}\nonumber
  & \F^{\rm BCS}(\Gamma) - \F^{\rm BCS}(\Gamma_0) \\ \nonumber & =
  -\frac T 2 \Tr\left[ \ln(1+e^{-\beta H_\Delta})-\ln(1+e^{-\beta
      H_0})\right] \\ \nonumber & \quad - h^{2-2d}\int_{\calC\times
    \R^d} V(h^{-1}(x-y)) \tfrac 14 \left| \psi_<(x)+\psi_<(y)\right|^2
  |\alpha_0(h^{-1}(x-y))|^2\, \frac{dx\,dy}{(2\pi)^d}\\ & \quad +
  \tfrac 12 T \, \H(\Gamma,\Gamma_\Delta) + \int_{\calC\times \R^d}
  V(h^{-1}(x-y))|\sigma(x,y)|^2\, {dx\,dy}\,, \label{off}
\end{align}
where $\H$ denotes again the relative entropy defined in
(\ref{relent}).

The terms in the first two lines on the right side of (\ref{off}) have
already been calculated. The first term is estimated in
Theorem~\ref{thm:scl}, and a bound on the second term was derived in
Section~\ref{sec:up} on the upper bound. The error in replacing
$\beta$ with $\beta_c$ is as for the upper bound. Using the uniform
upper bound on the $H^1$ norm of $\psi_<$, as well as
$\|\psi_<\|_{H^2} \leq C \epsilon/h$, we obtain the lower bound
\begin{align}\nonumber
  \F^{\rm BCS}(\Gamma) - \F^{\rm BCS}(\Gamma_0) &\geq h^{4-d} \left(
    \E^{\rm GL}(\psi_<) - B_3 - C ( h+\epsilon^2) \right) \\ & \quad +
  \tfrac 12 T \, \H(\Gamma,\Gamma_\Delta) + \int_{\calC\times \R^d}
  V(h^{-1}(x-y))|\sigma(x,y)|^2\, {dx\,dy}\,. \label{lb2:s1}
\end{align}
It remains to show that the terms in the last line of (\ref{lb2:s1})
are negligible, i.e., of higher order than $h^{4-d}$, for an
appropriate choice of $\epsilon \ll 1$. This will be accomplished in
the next step.

\subsection{Step 2}

We again employ Lemma~\ref{lem:klein} to get a lower bound on the
relative entropy $\H(\Gamma,\Gamma_\Delta)$. It implies that
\begin{equation}\label{eqs2}
  T \, \H(\Gamma,\Gamma_\Delta) \geq   \Tr\left[ \frac{H_\Delta}{\tanh \tfrac 12 \beta H_\Delta}\left( \Gamma-\Gamma_\Delta\right)^2 \right] \,.
\end{equation}
As in the proof of Lemma~\ref{lem:mon}, we shall use the fact that $x
\mapsto \sqrt x/ \tanh \sqrt x$ is an operator monotone function.

An application of Schwarz's inequality yields
\begin{equation}
  H_\Delta^2 \geq (1-\eta) H_0^2 - \eta^{-1}  \|\Delta\|_\infty^2 
\end{equation}
for any $\eta>0$.  To bound $H_0^2$ we proceed as in the proof of
Lemma~\ref{lem:mon}, specifically using (\ref{sss}) and
(\ref{boundq}), which states that $[ -i\nabla\, A(x)- i A(x)\, \nabla
]^2\leq C (1-\nabla^2)$. The choice $\epsilon = O(h)$ in (\ref{sss})
yields the lower bound
\begin{equation}
  H_0^2 \geq  (1-O(h)) [-h^2\nabla^2 - \mu]^2\otimes \id_{\C^2} - C h \,.
\end{equation}
The operator monotonicity thus implies that
\begin{align}\nn
  K_T^{0,0}\otimes \id_{\C^2} &\leq \frac { (1-\eta-O(h))^{-1/2}
    \sqrt{H_\Delta^2 + \eta^{-1}\|\Delta\|_\infty^2+Ch}}{\tanh\left[
      \tfrac 12 \beta (1-\eta-O(h))^{-1/2} \sqrt{H_\Delta^2 +
        \eta^{-1}\|\Delta\|_\infty^2 + Ch}\right]} \\ \nn & \leq
  (1-\eta-O(h))^{-1/2} \frac { \sqrt{H_\Delta^2 +
      \eta^{-1}\|\Delta\|_\infty^2+Ch)}}{\tanh\left[\tfrac 12 \beta
      \sqrt{H_\Delta^2 + \eta^{-1}\|\Delta\|_\infty^2+Ch}\right]} \\ &
  \leq (1-\eta-O(h))^{-1/2} \left( 1+ \tfrac 14 \beta^2
    \left(\eta^{-1} \|\Delta\|_\infty^2 + Ch \right)\right) \frac
  {H_\Delta}{\tanh \tfrac 12 \beta H_\Delta}
\end{align}
for $0<\eta<1$. With the choice $\eta = O(\|\Delta\|_\infty)$ this
gives
\begin{equation}
  \frac {H_\Delta}{\tanh \tfrac 12 \beta H_\Delta} \geq (1-O(h+\|\Delta\|_\infty))  K_T^{0,0} \otimes \id_{\C^2}\,.
\end{equation}
In particular, we infer from (\ref{eqs2}) that
\begin{equation}\label{eqs3}
  \frac T2  \, \H(\Gamma,\Gamma_\Delta) \geq    (1-O(h+\|\Delta\|_\infty)) \, \Tr\left[ K_T^{0,0} (\alpha-\alpha_\Delta)(\bar\alpha-\bar\alpha_\Delta) \right]\,,
\end{equation}
where $\alpha_\Delta$ denotes again the upper off-diagonal entry of
$\Gamma_\Delta$.

From the definition of $\Delta$, we see that
\begin{equation}
  \|\Delta\|_\infty \leq  h \|\psi_<\|_\infty \| t \|_\infty \,.
\end{equation}
Moreover, since the Fourier transform of $\psi_<$ is supported in the
ball $|p|\leq \epsilon/h$,
\begin{align}\nn
  \|\psi_<\|_\infty \leq \sum_p |\widehat\psi_<(p)| & \leq
  \|\psi_<\|_{H^1(\calC)} \left( \sum_{|p|\leq \epsilon h^{-1}} \frac
    1{1+p^2} \right)^{1/2}\\ & \leq C\times \left\{ \begin{array}{cl}
      1 & \text{for $d=1$} \\ \sqrt{ \ln (\epsilon/h)} & \text{for
        $d=2$} \\ \sqrt{\epsilon/h} & \text{for
        $d=3$.} \end{array}\right.
\end{align}

Recall the decomposition (\ref{defpsi}) of $\alpha$, and define $\phi$
by
\begin{equation}
  \alpha_\Delta = \tfrac h2 \left( \psi_<(x) \widehat\alpha_0(-ih\nabla) + \widehat \alpha_0(-ih\nabla) \psi_<(x)\right) +\phi \,.
\end{equation}
We thus have
\begin{equation}
  \alpha - \alpha_\Delta = \sigma  - \phi \,.
\end{equation}
Since $\|\psi_<\|_{H^2}\leq O(\epsilon/h)$, Theorem~\ref{lem3} implies
that $\|\phi\|_{H^1} \leq O(\epsilon h^{2-d/2})$.  From the positivity
of $K_T^{0,0}$ we conclude that
\begin{equation}\label{reml}
  \Tr\left[ K_T^{0,0} (\sigma -\phi) (\bar \sigma-\bar\phi) \right] \geq  \Tr  K_T^{0,0} \sigma \bar \sigma   - 2 \re\, \Tr  K_T^{0,0} \bar\sigma \phi\,. 
\end{equation}

The terms quadratic in $\sigma$ are thus
\begin{equation}
  (1- \delta )  \Tr  K_T^{0,0} \sigma \bar \sigma + \int_{\calC\times \R^d} V(h^{-1}(x-y))|\sigma(x,y)|^2\,{dx\,dy} 
\end{equation}
with $\delta=O(h+\|\Delta\|_\infty)$.  Pick some $\widetilde \delta
\geq 0 $ with $\delta + \widetilde \delta \leq 1/2$, and write
\begin{align}\nn
  (1- \delta ) K_T^{0,0} + V & = \widetilde \delta K_T^{0,0} + \left(
    1 - 2 \delta - 2 \widetilde \delta\right) \left( K_T^{0,0} + V
  \right) + \left(\delta+\widetilde \delta\right) \left(K_T^{0,0} + 2
    V \right) \\ & \geq \widetilde \delta K_T^{0,0} - 2 D T_c h^2 - C
  \left( \delta + \widetilde \delta\right)\,,
\end{align}
where we have used that $V$ is relatively form-bounded with respect to
$K_T^{0,0}$ to bound the last term, and $K_T^{0,0}\geq K_{T_c}^{0,0} -
2 (T_c - T)= K_{T_c}^{0,0} - 2 h^2 D T_c$ to bound the second. Using
also that $K_T^{0,0} \geq -h^2 \nabla^2 -\mu$, we thus conclude that
\begin{align}\nn
  & (1-\delta) \Tr K_T^{0,0} \sigma \bar \sigma + \int_{\calC\times
    \R^d} V(h^{-1}(x-y))|\sigma(x,y)|^2\,{dx\,dy} \\ & \geq \widetilde
  \delta \left( \| \sigma\|_{H^1}^2 - (C+\mu+1) \|\sigma\|_2^2\right)
  - \left( C \delta + 2D T_c h^2 \right) \|\sigma\|_2^2\,. \label{tdt}
\end{align}
Recall that $\|\sigma\|_2\leq O(h^{2-d/2}/\epsilon)$. We shall choose
$\widetilde \delta = 0$ if the first parenthesis on the right side of
(\ref{tdt}) is less than $\tfrac 12 \|\sigma\|_{H^1}^2$ (and, in
particular, if it is negative), while $\widetilde \delta = O(1)$ in
the opposite case, i.e., when $\|\sigma\|_{H_1}^2 \geq 2
(C+\mu+1)\|\sigma\|_2^2$. In the latter case we shall have the
positive term $\widetilde\delta \|\sigma\|_{H^1}^2/2$ at our disposal,
which will be used in (\ref{upt}) below.

We are left with estimating the last term in (\ref{reml}). It can be
bounded by the product of the $H^1$ norms of $\sigma$ and $\phi$. This
turns out not to be good enough, however.  Recall from
(\ref{def:sigma}) that $\sigma$ is a sum of two terms, $\xi$ and
$\sigma - \xi$, where the latter is proportional to $\psi_>$, and
$\|\xi\|_{H^1} \leq O (h^{2-d/2})$.  Moreover, as the proof of
Theorem~\ref{lem3} in Section~\ref{sec:propproof} shows, $\phi$ is the
sum of two terms, $\eta_1$ and $\phi-\eta_1$, with $\eta_1$ defined in
(\ref{def:eta1}) (with $\psi$ replaced by $\psi_<$) and
$\|\phi-\eta_1\|_{H^1} \leq O(h^{3-d/2})$.  Now
\begin{equation}
  \Tr  K_T^{0,0} \left( \bar \sigma - \bar \xi\right) \eta_1 = 0
\end{equation}
as can be seen by writing out the trace in momentum space and using
that $\widehat \psi_<$ and $\widehat \psi_>$ have disjoint support.
Hence
\begin{align}\nn
  \re\, \Tr K_T^{0,0} \bar \sigma \phi & \leq C \left( \|\xi \|_{H^1}
    \|\phi\|_{H^1} + \|\sigma\|_{H^1} \|\phi -\eta_1\|_{H^1}\right)\\
  & \leq O(\epsilon h^{4-d}) + O(h^{3-d/2})\|\sigma\|_{H^1}\,.
\end{align}
In the case $\|\sigma\|_{H^1}\leq C \|\sigma\|_2$ (corresponding to
$\widetilde \delta =0$ above) we can further bound
$\|\sigma\|_{H^1}\leq O(h^{2-d/2}/\epsilon)$.  In the opposite case,
where $\widetilde \delta = O(1)$, we can use the positive term
$\widetilde\delta \|\sigma\|_{H^1}^2/2$ from before and bound
\begin{equation}\label{upt}
  \frac {\widetilde \delta}{2}  \|\sigma\|_{H^1}^2 -  O(h^{3-d/2})\|\sigma\|_{H^1}  \geq  - O(h^{6-d})\,,
\end{equation}
which thus leads to an even better bound.

In combination with (\ref{lb2:s1}) these bounds show that
\begin{equation}\label{im1}
  h^{d-4} \left( \F^{\rm BCS}(\Gamma) - \F^{\rm BCS}(\Gamma_0) \right) \geq \E^{\rm GL}(\psi_<) - B_3 - C e
\end{equation}
where
\begin{equation}
  e =  h+\epsilon^2+ \epsilon + \frac h\epsilon + \frac h{\epsilon^2} \times \left\{ \begin{array}{cl} 1 & \text{for $d=1$} \\ \sqrt{ \ln (\epsilon/h)} & \text{for $d=2$} \\ \sqrt{\epsilon/h} & \text{for $d=3$.} \end{array}\right.
\end{equation}
The choice $\epsilon = h^{1/3}$ for $d=1$, $\epsilon = h^{1/3}
[\ln(1/h)]^{1/6}$ for $d=2$ and $\epsilon = h^{1/5}$ for $d=3$,
respectively, leads to
\begin{equation}
  e \leq C  \times \left\{ \begin{array}{cl} h^{1/3} & \text{for $d=1$} \\ h^{1/3} \left[\ln (1/h)\right]^{1/6} & \text{for $d=2$} \\ h^{1/5} & \text{for $d=3$.} \end{array}\right.
\end{equation}

The completes the lower bound to the BCS energy. The statement
(\ref{thm:dec}) about the minimizer follows immediately from
(\ref{im1}) and (\ref{defpsi}).

\section{Proof of Theorem~\ref{thm:scl}}\label{sec:thm:scl}

For simplicity, we prove Theorem~\ref{thm:scl} only in the case
$d=3$. The cases $d=1$ and $d=2$ are very similar and are left to the
reader.

Note that the function $f$ in (\ref{deff}) is real analytic and has a
bounded derivative. The second and higher derivatives decay
exponentially. In particular, the following lemma applies. It will be
used repeatedly in the proof of Theorem~\ref{thm:scl}.

For a general smooth function $f:\R\to\R$, let $[a_1,\dots,a_N]_f$
denote the divided differences \cite{dono}, defined recursively via
\begin{equation}\label{def:re} [a]_f = f(a) \ , \quad [a_1,a_2]_f =
  \frac{f(a_1)-f(a_2)}{a_1-a_2}
\end{equation}
and
\begin{equation}\label{def:rec} [a_1,a_2,\dots,a_N]_f = \frac{
    [a_1,\dots,a_{N-1}]_f - [a_2,\dots,a_{N}]_f}{a_1-a_N}
\end{equation}
is case $a_1\neq a_N$. The extension to coinciding arguments is simply
by continuity. In case of distinct arguments,
\begin{equation} [a_1,a_2,\dots,a_N]_f = \sum_{j=1}^N f(a_j)
  \prod_{i,i\neq j} (a_j-a_i)^{-1}\,.
\end{equation}
In case $f$ is analytic in a neighborhood of the real axis, we have
\begin{equation}\label{83}
  \left[a_1,\dots,a_N\right]_f = \frac 1{2\pi i} \int_\Gamma f(z) \prod_{k=1}^N \frac 1{z-a_k} dz\,,
\end{equation}
with $\Gamma$ a contour enclosing all the $a_i$.

\begin{Lemma}\label{lem:dec}
  Let $f$ be a smooth function on $\R$ satisfying $|f^{(n)}(x)| \leq C
  (1+|x|)^{1-n}$ for $1\leq n\leq N-1$.  Then, for $N\geq 3$, there is
  finite constant $C_N>0$ such that
  \begin{equation}\label{eq:lem:dec1}
    \left[ a_1,\dots,a_N\right]_f \leq \frac {C_N}{1+ \max_i\{|a_i|\}} \quad \text{for all $a_i\in \R$.}
  \end{equation}
  Moreover, for all $1\leq n < N$ there is a finite constant $C_N'>0$
  such that if, $a_i \leq -\lambda \leq 0$ for all $1\leq i\leq n$ and
  $a_i\geq \lambda\geq 0$ for all $n+1\leq i\leq N$, then
  \begin{equation}\label{eq:lem:dec2}
    \left[ a_1,\dots,a_N\right]_f \leq \frac {C'_N}{(1+\lambda)^{N-2}} \,.
  \end{equation}
\end{Lemma}

\begin{proof}
  We first prove (\ref{eq:lem:dec2}). Using Feynman's formula
  \begin{align}\nonumber
    \left[a_1,\dots,a_N\right]_f & = (N-1)! \int_{[0,1]^N} \left[
      \mbox{$\sum_{i=1}^N$} c_i a_i , \dots,\mbox{$\sum_{i=1}^N$} c_i
      a_i\right]_f \delta\left(1-\mbox{$\sum_{i=1}^N$} c_i\right)
    \prod_{k=1}^N dc_k \\ &= \int_{[0,1]^N}
    f^{(N-1)}\left(\mbox{$\sum_{i=1}^N$} c_i a_i\right)
    \delta\left(1-\mbox{$\sum_{i=1}^N$} c_i\right) \prod_{k=1}^N
    dc_k \label{feynman}
  \end{align}
  we see that $[a_1,\dots,a_N]$ is uniformly bounded for $N\geq
  2$. Hence it suffices to consider the case $\lambda \geq
  1$. Similarly, a simple change of variables shows that
  \begin{align}\nn
    & \left[a_1,\dots,a_N\right]_f \\ \nn & = \int_0^1 \int_{[0,1]^N}
    f^{(N-1)}\left(\gamma \, \mbox{$\sum_{i=1}^n$} c_i a_i +
      (1-\gamma) \mbox{$\sum_{i=n+1}^N$} c_i a_i \right)\\ & \qquad
    \times \gamma^{n-1} (1-\gamma)^{N-n-1}
    \delta\left(1-\mbox{$\sum_{i=1}^n$}
      c_i\right)\delta\left(1-\mbox{$\sum_{i=n+1}^N$} c_i\right)
    \prod_{k=1}^N dc_k \,d\gamma \label{feynman2}
  \end{align}
  and hence it is sufficient to consider the case $a_1=\dots=a_n=a$
  and $a_{n+1}=\dots=a_N=b$. In this special case, we have
  \begin{equation}
    \left[a,\dots,a,b,\dots, b\right]_f = 
    \frac 1{(n-1)!(m-1)!} \left(\frac{\partial}{\partial a}\right)^{n-1}\left(\frac{\partial}{\partial b}\right)^{m-1}\frac{f(a)-f(b)}{a-b} \,,
  \end{equation}
  where $m=N-n$. Note that $a\leq -\lambda \leq -1$ and $b\geq
  \lambda\geq 1$. The result now follows easily from our assumptions
  on $f$.

  Let $\kappa = \max_i |a_i|$. To prove (\ref{eq:lem:dec1}), we may
  again assume that $\kappa\geq 1$. If $\max_i a_i - \min_i a_i \leq
  \kappa/2$, then either $a_i\geq \kappa/2$ for all $i$ or $a_i\leq
  -\kappa/2$ for all $i$, and the result follows from
  (\ref{feynman}). If, on the other hand, $\max_i a_i -\min_i a_i \geq
  \kappa/2$, the result follows immediately from the definition
  (\ref{def:rec}), with $\min_i a_i$ in place of $a_1$ and $\max_i
  a_i$ in place of $a_N$, using the boundedness of the numerator.
\end{proof}

What the lemma says is that $[a_1,a_2,\dots,a_N]_f$ decays at least as
fast as the inverse of its largest argument. If all the arguments get
large, but in opposite directions, than the decay is at least as fast
as the $(N-2)$th power of the inverse.

Let now $f$ denote the function defined in (\ref{deff}). It is
analytic in the strip $|\im z|< \pi$. This leads to the following
contour integral representation.

\begin{Lemma}\label{lem:int}
  For $R>0$, let $\Gamma_R$ be the contour $\{r - i \tfrac
  \pi{2\beta},\, r\in [-R,R] \}\cup \{-r + i \tfrac \pi{2\beta},\,
  r\in [-R,R] \}$. Then
  \begin{equation}\label{int:rep}
    f( \beta H_\Delta)-f(\beta H_0)  = \lim_{R\to \infty}  \frac 1{2\pi i}\int_{\Gamma_R} f(\beta z) \, \left( \frac 1{z-H_\Delta} - \frac 1{z-H_0} \right)\, dz
  \end{equation}
  where the limit holds in the weak sense, i.e., for expectation
  values with functions in $C_0^\infty(\R^d)\oplus C_0^\infty(\R^d)$.
\end{Lemma}

Since $f$ is unbounded, it is important to take the
difference of the two operators $f(\beta H_\Delta)$ and $f(\beta
H_0)$. The representation (\ref{int:rep}) would not hold for the
individual terms alone.  We shall write, for simplicity,
\begin{equation}\label{sp2}
  f( \beta H_\Delta)-f(\beta H_0)  = \frac 1{2\pi i}\int_\Gamma f(\beta z) \, \left( \frac 1{z-H_\Delta} - \frac 1{z-H_0} \right)\, dz
\end{equation}
where $\Gamma$ is the contour $z= r \pm i\tfrac \pi {2\beta}$, $r\in
\R$. This equality has to be understood as the weak limit
(\ref{int:rep}).

\begin{proof}
  For fixed $\lambda \in \R$, we have
  \begin{equation}
    \lim_{R\to \infty}  \frac 1{2\pi i}\int_{\Gamma_R} f(\beta z) \,  \frac 1{z - \lambda} \, dz = f(\beta\lambda) - \frac 12\,,
  \end{equation}
  and the limit is uniform on bounded sets in $\lambda$.  For
  functions $\phi$ in the range of $\theta(\lambda - |H_\Delta|)$ for
  some $\lambda>0$, this implies that
  \begin{equation}
    \lim_{R\to \infty}  \frac 1{2\pi i}\int_{\Gamma_R} f(\beta z) \,  \left\langle \phi \left| (z - H_\Delta)^{-1} \right| \phi\right\rangle  \, dz =  \left\langle \phi \left | f(\beta H_\Delta) - \tfrac 12 \right| \phi\right\rangle \,.
  \end{equation}
  Since the operators $\int_{\Gamma_R} f(\beta z) (z-H_\Delta)^{-1}dz$
  are bounded, in absolute value, by $C(1+|H_\Delta|)$, uniformly in
  $R$, the limit extends to all functions $\phi$ in the form domain of
  $1+|H_\Delta|$. The same argument applies with $H_0$ in place of
  $H_\Delta$. The form domains of $1+|H_\Delta|$ and $1+|H_0|$ are
  equal since $\Delta$ is bounded, and contain $C_0^\infty(\R^d)\oplus
  C_0^\infty(\R^d)$. This completes the proof.
\end{proof}

The diagonal entries of the $2\times 2$ matrix-valued operator $f(
\beta H_\Delta)-f(\beta H_0) $ are complex conjugates of each
other. This follows from the unitary equivalence (\ref{eq:unit}) and
the fact that $f(-z) = f(z) - z$. The unitary $U$ in (\ref{eq:unit})
interchanges the diagonal entries, hence the upper left entry of $f(
\beta H_\Delta)-f(\beta H_0)$ equals the lower right entry of
\begin{align}\nonumber
  U f(\beta H_\Delta) U^\dagger - U f(\beta H_0)U^\dagger & = f(-
  \beta \bar H_\Delta ) - f(-\beta \bar H_0) \\ & \nonumber = f( \beta
  \bar H_\Delta) - f(\beta \bar H_0) - \beta(\bar H_\Delta - \bar H_0)
  \\ & = \overline{ f( \beta H_\Delta)} -\overline{ f(\beta H_0)} -
  \beta(\bar H_\Delta - \bar H_0) \,.
\end{align}
Since the diagonal entries of $H_\Delta - H_0$ are zero, this proves
the claim.  In particular, the diagonal entries of $f( \beta
H_\Delta)-f(\beta H_0) $ have the same trace, and hence it suffices to
study the upper left diagonal entry, which we denote by
$[\,\cdot\,]_{11}$.

Let $k$ denote the operator
\begin{equation}\label{def:k}
  k = \left(-ih\nabla + h A(x) \right)^2 - \mu - h^2 W(x) \,.
\end{equation}
The resolvent identity and the fact that
\begin{equation}
  \delta:=H_\Delta-H_0 = \left(\begin{array}{cc} 0 & \Delta \\ \bar\Delta  & 0 \end{array}\right)
\end{equation}
is off-diagonal (as an operator-valued $2\times 2$ matrix) implies
that
\begin{equation}
  \left[ \frac 1{z-H_\Delta} - \frac 1{z-H_0} \right]_{11} = I_1 + I_2 + I_3
\end{equation}
where
\begin{equation}\label{def:i1}
  I_1 =  \frac 1{z-k} \Delta \frac 1{z+\bar k} \Delta^\dagger \frac 1{z-k} \,,
\end{equation}
\begin{equation}\label{def:i2}
  I_2  =   \frac 1{z-k} \Delta \frac 1{z+\bar k} \Delta^\dagger \frac 1{z-k} \Delta \frac 1{z+\bar k} \Delta^\dagger \frac 1{z-k} \,,
\end{equation}
and
\begin{equation}\label{def:i3}
  I_3 =  \left[\frac 1 {z-H_\Delta} \right]_{11} \Delta \frac 1{z+\bar k} \Delta^\dagger \frac 1{z-k} \Delta \frac 1{z+\bar k} \Delta^\dagger \frac 1{z-k}\Delta \frac 1{z+\bar k} \Delta^\dagger \frac 1{z-k} \,.
\end{equation}
We shall first give a bound on the contribution of $I_3$ to the
integral in (\ref{sp2}).

\begin{Lemma}\label{lem8}
  \begin{equation}
    \left\| \int_\Gamma f(\beta z)\, I_3 \, dz\right\|_1  \leq C h^{6-d}  \|\psi\|_{H^1(\calC)}^6 \|t\|_6^6\,.
  \end{equation}
\end{Lemma}

\begin{proof}
  By the triangle inequality
  \begin{equation}
    \left\| \int_\Gamma f(\beta z)\, I_3 \, dz\right\|_1 \leq \int_\Gamma |f(\beta z)| \left\| I_3\right\|_1 |dz| \,.
  \end{equation}
  Using H\"older's inequality (\ref{genholder}) and $|z-H_\Delta|\geq
  \pi/(2\beta)$ for $z\in\Gamma$, we obtain the bound
  \begin{equation}
    \|I_3\|_1 \leq \frac {4\beta} \pi \|\Delta\|_6^6 \left\|(z-k)^{-1}\right\|_\infty^3 \left\|(z+\bar k)^{-1}\right\|_\infty^3\,.
  \end{equation}
  Moreover, since $k\geq -\mu -h^2 \|W\|_\infty$,
  \begin{equation}\label{infb}
    \left\|(z-k)^{-1}\right\|_\infty \leq C \times \left\{ \begin{array}{cl} 1  & \text{for $r\gg 1$} \\ O(|r|^{-1}) & \text{for $r\ll -1$} \end{array}\right.
  \end{equation}
  for $z=r\pm i\pi/(2\beta)$, and hence
  \begin{equation}
    \|I_3\|_1 \leq \frac{C}{1+|z|^3} \|\Delta\|_6^6 \,.
  \end{equation}
  It is important to get a decay faster than $|z|^{-2}$ since we are
  integrating against the function $|\ln(1+e^{-\beta z})|$ which grows
  linearly as $z\to -\infty$. We conclude that
  \begin{equation}
    \int_\Gamma |f(\beta z)| \,\left\|   I_3 \right\|_1 \, |dz|  \leq C  \|\Delta\|_6^6\,.
  \end{equation}
  We can further bound $\|\Delta\|_6 \leq (2\pi)^{-d/6}\|\psi\|_{6}
  h^{1-d/6} \|t\|_6$, according to the triangle inequality and
  (\ref{lte}). Finally, $\|\psi\|_6 \leq C \|\psi\|_{H^1(\calC)}$ for
  $d\leq 3$, by Sobolev's inequality for functions on the torus. This
  completes the proof.
\end{proof}

We continue with the bound on $I_2$. Recall the definition of the
divided differences (\ref{def:re})--(\ref{def:rec}).

\begin{Lemma}\label{lem:aidef} The operator $\int_\Gamma f(\beta z)
  I_2 dz$ is locally trace class, and
  \begin{align}\nn
    &\Biggl| \frac 1{2\pi i} \Tr \int_\Gamma f(\beta z) \, I_2 \, dz -
    \\ \nn & \ h^{4-d} \sum_{p_1,p_2,p_3 \in (2\pi \Z)^d} \widehat
    \psi(p_1) \widehat \psi^*(p_2) \widehat\psi(p_3) \widehat
    \psi^*(-p_1-p_2-p_3) F(hp_1,hp_2,hp_3) \Biggl| \\ & \leq C h^{5-d}
    \|t\|_6^4 \|\psi\|^4_{H^1(\calC)} \label{cg6}
  \end{align}
  where the function $F$ is given by
  \begin{align}\nonumber
    & F(p_1,p_2,p_3) \\ &= \beta^4\int_{\R^d} \prod_{i=1}^4 \frac{
      t(q+\sum_{j=1}^{i-1} p_i) + t(q+\sum_{j=1}^{i}p_j)}{2}
    \left[a_3,a_3,a_1,-a_2,-a_0\right]_f \, \frac{dq}{(2\pi)^d}
    \label{defF3}
  \end{align}
  and where we introduced the notation $a_i = \beta ( (q+\sum_{j=1}^i
  p_i)^2-\mu)$ and $p_4:=-p_1-p_2-p_3$.
\end{Lemma}

One can actually show that (\ref{cg6}) holds with $h^{6-d}$ in place
of $h^{5-d}$ in the error term. This requires substantially more
effort, however, and hence we omit the proof for the sake of brevity.

\begin{proof}
  The resolvent identity reads
  \begin{equation}\label{res}
    \frac 1{z-k} = \frac 1{z-k_0} + \frac 1{z-k_0} \left( k - k_0  \right) \frac 1{z-k} \,, 
  \end{equation}
  where
  \begin{equation}\label{def:k0}
    k_0 = -h^2 \nabla^2 - \mu\,,
  \end{equation}
  and hence
  \begin{equation}\label{difk}
    k-k_0 =  h^2 W(x) + h^2 A(x)^2 - i h^2 \nabla\cdot A(x) - i h^2 A(x)\cdot \nabla \,.
  \end{equation}
  We apply this to the first factor in $I_2$ in (\ref{def:i2}). Using
  again the H\"older inequality (\ref{genholder}), we can bound
  \begin{align}\nonumber
    & \left\| \left(\frac 1{z-k}-\frac 1{z-k_0}\right) \Delta \frac
      1{z+\bar k} \Delta^\dagger \frac 1{z-k} \Delta \frac 1{z+\bar k}
      \Delta^\dagger \frac 1{z-k} \right\|_1 \\ & \leq\left\|\frac
      1{z-k_0} \left( k-k_0\right) \right\|_\infty \|\Delta\|_6^4
    \|(z-k)^{-1}\|_\infty^3 \|(z+\bar k)^{-1}\|_6^2 \,. \label{f1}
  \end{align}

  Using our assumptions on $A$ and $W$, it is easy to see that, for $z
  = r \pm i \pi/(2\beta)$,
  \begin{equation}\label{resb1}
    \left\|\frac 1{z-k_0} \left(- i h^2 \nabla A(x) - i h^2 A(x)\nabla\right) \right\|_\infty \leq C h \times  \left\{ \begin{array}{cl} \sqrt{1+|r|}  & \text{for $r\gg 1$} \\ \frac 1{\sqrt{1+|r|}} & \text{for $r\ll -1$} \end{array}\right. 
  \end{equation}
  and also
  \begin{equation}\label{resb2}
    \left\|\frac 1{z-k_0} \left( h^2 W(x) + h^2 A(x)^2\right) \right\|_\infty \leq C h^2 \times  \left\{ \begin{array}{cl} 1  & \text{for $r\gg 1$} \\ \frac 1{1+|r|} & \text{for $r\ll -1$.} \end{array}\right. 
  \end{equation}
  As in the proof of Lemma~\ref{lem8}, we can use (\ref{infb}) to
  bound $\|(z-k)^{-1}\|_\infty$. Moreover, it is not difficult to see
  that
  \begin{equation}\label{pb}
    \left\|(z-k)^{-1}\right\|_p \leq C\, h^{-d/p} \times \left\{ \begin{array}{cl} r^{(d-2)/(2p)}  & \text{for $r\gg 1$} \\ |r|^{-1+d/(2p)} & \text{for $r\ll -1$} \end{array}\right.
  \end{equation}
  for $d/2<p\leq \infty$, generalizing (\ref{infb}). This follows from
  noting that $k$ can be replaced by $k_0$ for a bound, since their
  spectrum agrees up to $o(1)$. With $k_0$ in place of $k$, the result
  follows from evaluating the corresponding integral. In case $k$ is
  replaced by $-\bar k$, a similar bound holds, with the estimates for
  $r\gg 1$ and $r\ll -1$ interchanged.

  For $d=3$, we hence end up with a function that decays like
  $|r|^{-10/3}$ for negative $r$ and $r^{-1}$ for positive $r$. Since
  we integrate against a function that decays exponentially for
  positive $r$ and increases linearly for negative $r$, the remaining
  contour integral is finite. We conclude that (\ref{f1}), multiplied
  by $f(\beta z)$ and integrated over $\Gamma$, is bounded by $C
  h^{1-d/3}\|\Delta\|_6^4$. As in the proof of Lemma~\ref{lem8}, we
  can bound $\|\Delta\|_6 \leq C \|\psi\|_{H^1(\calC)} h^{1-d/6}
  \|t\|_6$. We have thus obtained a bound on the error made by
  replacing the first factor $(z-k)^{-1}$ in $I_2$ by $(z-k_0)^{-1}$.

  In exactly the same way we proceed with the remaining factors
  $(z-k)^{-1}$ and $(z+\bar k)^{-1}$ in $I_2$. The only difference is
  that $k$ might now be replaced by $k_0$ in the factors we have
  already treated, but this does not affect the bounds. Also the
  number of $+$ and $-$ terms is different, but we still get a decay
  of at least $|r|^{-7/3}$ for negative $r$, which is sufficient for
  the contour integral to converge.

  The final result is that
  \begin{align}\nn
    & \left\| \int_\Gamma f(\beta z) \left[ I_2 - \frac 1{z-k_0}
        \Delta \frac 1{z+k_0} \Delta^\dagger \frac 1{z-k_0} \Delta
        \frac 1{z+k_0} \Delta^\dagger \frac 1{z-k_0} \right] \, dz
    \right\|_1 \\ & \leq O( h^{5-d}) \|\psi\|^4_{H^1(\calC)}
    \,. \label{i2rem}
  \end{align}
  Proceeding as above, one shows that the remaining term on the left
  side of (\ref{i2rem}) is trace class.  An explicit computation of
  the trace per unit volume yields the desired result.
\end{proof}

Let us now look at the behavior of $F(p_1,p_2,p_3)$ for small
arguments. We have
\begin{equation}
  \left[a,a,a,-a,-a\right]_f  = \frac 1{16 a} g_1(a)
\end{equation}
with $g_1$ defined in (\ref{defg1}). With $a=\beta(q^2-\mu)$ we thus
have
\begin{equation}\label{f000}
  F(0,0,0) =  \frac{\beta^3}{16} \int_{\R^d} t(q)^4 \, \frac{g_1(\beta(q^2-\mu))}{q^2-\mu} \, \frac{dq}{(2\pi)^d} \,.
\end{equation}

\begin{Lemma} 
  For some constant $C$ depending on the $L^6$ norm of $t$ and its
  derivatives up to order four,
  \begin{equation}
    \left| F(p_1,p_2,p_3) - F(0,0,0)\right| \leq C \left( p_1^2 + p_2^2 + p_3^2\right)\,.
  \end{equation}
\end{Lemma}

\begin{proof}
  We will first show that $F$ is bounded. Using H\"older's inequality,
  we can bound
  \begin{equation}
    |F(p_1,p_2,p_3)| \leq \beta^4 \prod_{i=1}^4 \frac{S_i + S_{i+1}}{2}
  \end{equation}
  where $S_5:= S_1$ and
  \begin{equation}
    S_i^4 =  \int_{\R^d}  \left| t(q+\mbox{$\sum_{j=1}^{i-1}$}p_i)\right|^4 \left| \left[a_3,a_3,a_1,-a_2,-a_0\right]_f  \right| \frac{dq}{(2\pi)^d}
  \end{equation}
  for $1\leq i\leq 4$.  The last factor in the integrand is bounded by
  a constant times $(1+ \beta (q+\sum_{j=1}^{i-1} p_i)^2)^{-1}$, for
  any $1\leq i \leq 4$, as the bound (\ref{eq:lem:dec1}) from
  Lemma~\ref{lem:dec} shows.  A further application of H\"older's
  inequality thus shows that $S_i$ is bounded by $\|t\|_6$.

  It is thus sufficient to consider the case $|p_i|\leq C$ for all
  $i\in\{1,2,3\}$. We write
  \begin{equation}
    t(q+p) = t(q) + p\cdot \nabla t(q) + \int_{0}^1 (p\cdot \nabla)^2 t(q+ s p) (1-s) ds
  \end{equation}
  and also, using partial fractions (or, equivalently,
  (\ref{def:rec})),
  \begin{align}\nn
    &\left[ \beta((q+p)^2-\mu), a,b,c,d\right]_f \\ & = \left[\beta(
      q^2-\mu), a,b,c,d\right]_f + \beta (2p\cdot q + p^2) \left[
      \beta( (q+p)^2-\mu),\beta(q^2-\mu),
      a,b,c,d\right]_f\,.\label{iter}
  \end{align}
  We expand all these factors to second order in $p$ and, in
  particular, iterate (\ref{iter}) once more. Using the fact that the
  derivatives of $t$ are in $L^6$, as well as the decay bound
  (\ref{eq:lem:dec2}) of Lemma~\ref{lem:dec}, we can proceed as above,
  using H\"older's inequality, to arrive at the result. Note that all
  terms linear in $p$ integrate to zero since $t$ is
  reflection-symmetric.
\end{proof}

We have
\begin{equation}
  \sum_{p_1,p_2,p_3 \in (2\pi \Z)^d}  \widehat \psi(p_1) \widehat \psi^*(p_2) \widehat\psi(p_3) \widehat \psi^*(-p_1-p_2-p_3) = \int_\calC |\psi(x)|^4 dx \,.
\end{equation}
With the aid of the Schwarz inequality for the $p_1$ sum and Sobolev's
inequality, we can bound, similarly to (\ref{649})--(\ref{650}),
\begin{equation}\label{837}
  \sum_{p_1,p_2,p_3 \in (2\pi \Z)^d} p_1^2 \left| \widehat \psi^*(p_1) \widehat \psi^*(p_2) \widehat\psi(p_3) \widehat \psi(-p_1-p_2-p_3)\right| \leq C \|\psi\|_{H^2} \|\psi\|_{H^1}^3 \,.
\end{equation}
Equal bounds hold with $p_1^2$ replaced by $p_2^2$ and $p_3^2$,
respectively. Hence we arrive at the result
\begin{align}\nn
  & \left| \frac 1{2\pi i} \Tr \int_\Gamma f(\beta z)\, I_2 \, dz -
    h^{4-d}\, F(0,0,0) \, \int_\calC |\psi(x)|^4\,dx \right| \\ & \leq
  C h^{5-d} \|\psi\|_{H^1(\calC)}^4 + C h^{6-d} \ \|\psi\|_{H^2}
  \|\psi\|_{H^1}^3 \,,
\end{align}
with $F(0,0,0)$ given in (\ref{f000}). The term involving $F(0,0,0)$
gives rise to the last term in (\ref{def:e2}).

We are left with examining the contribution of $I_1$ to the integral
in (\ref{sp2}).

\begin{Lemma}\label{lem:i1a}
  The operator $\int_\Gamma f(\beta z) I_1 dz$ is locally trace class.
\end{Lemma}

The proof of Lemma~\ref{lem:i1a} is somewhat technical and will be
given in the appendix.

Using the resolvent identity (\ref{res}) we can write $I_1=
I_1^a+I_1^b+I_1^c+I_1^d$, where
\begin{equation}
  I_1^a=  \frac 1{z-k_0} \Delta \frac 1{z+k_0} \Delta^\dagger \frac 1{z-k_0}
\end{equation}
and
\begin{align}\nn
  I_1^b & = \frac 1{z-k_0} (k-k_0) \frac 1{z-k_0} \Delta \frac
  1{z+k_0} \Delta^\dagger \frac 1{z-k_0} \\ \nn & \quad + \frac
  1{z-k_0} \Delta \frac 1{z+k_0} (k_0-\bar k) \frac 1{z+k_0}
  \Delta^\dagger \frac 1{z-k_0} \\ \nn & \quad + \frac 1{z-k_0} \Delta
  \frac 1{z+k_0} \Delta^\dagger \frac 1{z-k_0} (k-k_0) \frac 1{z-k_0}
  \,.
\end{align}
The part $I_1^c$ consists of $6$ terms involving two $(k_0-k)$
factors. Explicitly,
\begin{align}\nn
  I_1^c & = \frac 1{z-k_0} (k-k_0) \frac 1{z-k_0}(k-k_0)\frac 1{z-k_0}
  \Delta \frac 1{z+k_0} \Delta^\dagger \frac 1{z-k_0} \\ \nn & \quad +
  \frac 1{z-k_0} \Delta \frac 1{z+k_0} (k_0-\bar k) \frac 1{z+k_0}
  (k_0-\bar k) \frac 1{z+k_0} \Delta^\dagger \frac 1{z-k_0} \\ \nn &
  \quad + \frac 1{z-k_0} \Delta \frac 1{z+k_0} \Delta^\dagger \frac
  1{z-k_0} (k-k_0) \frac 1{z-k_0} (k-k_0) \frac 1{z-k_0} \\ \nn &
  \quad + \frac 1{z-k_0} (k-k_0) \frac 1{z-k_0} \Delta \frac
  1{z+k_0}(k_0-\bar k) \frac 1{z+k_0} \Delta^\dagger \frac 1{z-k_0} \\
  \nn & \quad + \frac 1{z-k_0} (k-k_0) \frac 1{z-k_0} \Delta \frac
  1{z+k_0} \Delta^\dagger \frac 1{z-k_0}(k-k_0) \frac 1{z-k_0} \\ \nn
  & \quad + \frac 1{z-k_0} \Delta \frac 1{z+k_0} (k_0-\bar k) \frac
  1{z+k_0} \Delta^\dagger \frac 1{z-k_0} (k-k_0) \frac 1{z-k_0} \,.
\end{align}
Finally, $I_1^d$ contains all the rest.

A straightforward computation of the trace and the contour integral
shows that
\begin{equation}\label{trace:i1a}
  \frac 1{2\pi i} \Tr  \int_\Gamma f(\beta z) \,I_1^a \, dz = h^{2-d} \sum_{p \in (2\pi \Z)^d} | \widehat \psi(p) |^2  G (hp) \,,
\end{equation}
where
\begin{equation}\label{def:fG}
  G(p_1) =\beta^2 \int_{\R^d} \frac{(t(q+p_1)+t(q))^2}4  \left[a_1,a_1,-a_0\right]_f  \, \frac{dq}{(2\pi)^d}\,.
\end{equation}
Here, we use the notation introduced in Lemma~\ref{lem:aidef}.
Explicitly, $G(p)$ equals
\begin{equation}
  -\frac \beta{4}
  \int_{\R^d} \frac{\left( t(q+p) + t(q)\right)^2}{4}  \frac { \tanh\left(\tfrac 12 \beta((q+p)^2-\mu)\right) +   \tanh\left(\tfrac 12 \beta(q^2-\mu)\right)}{(q+p)^2 + q^2 -2\mu}\, \frac{dq}{(2\pi)^d}\,.
\end{equation}

\begin{Lemma}\label{lem11}
  The function $G$ in (\ref{def:fG}) is bounded, twice differentiable
  at zero, and
  \begin{equation}\label{gexp}
    \left| G( p) - G(0) - \tfrac 12 (p\cdot \nabla)^2 G(0) \right| \leq C |p|^4
  \end{equation}
  for some constant $C$ depending only on $\int_{\R^d} t(q)^2
  (1+q^2)^{-1} dq$ and the $L^6$ norm of $t$ and its derivatives up to
  order four.
\end{Lemma}

When inserted into (\ref{trace:i1a}), the term $C|p|^4$ thus yields an
error of the order $ h^{6-d} \|\psi\|_{H^2}^2$.

\begin{proof}
  It follows from (\ref{eq:lem:dec1}) of Lemma~\ref{lem:dec} that
  \begin{equation}
    \left| \left[a_1,a_1,-a_0\right]_f   \right| \leq  \frac C{1+ q^2 + (p+q)^2}\,,
  \end{equation}
  hence $G$ is bounded by $\int_{\R^d} t(q)^2 (1+q^2)^{-1} dq$, which
  is finite by our assumption (\ref{t:as2}) on $t$. To prove
  (\ref{gexp}), it thus suffices to consider the case $|p|\leq C$.

  We expand $t(q+p)$ up to fourth order in $p$. Similarly, we write
  $a_1 = a_0 + \beta(2p_1\cdot q + p_1^2)$ and expand the brackets,
  using the resolvent identity. The decay of the resulting expression
  for large $q$ is governed by (\ref{eq:lem:dec2}).  Using in addition
  that $t$ is reflection symmetric and satisfies (\ref{t:as1}), the
  result follows in a straightforward way from H\"older's
  inequality. We omit the details.
\end{proof}

Recall the definition of the $g_i$ in (\ref{defg0})--(\ref{defg2}). We
have
\begin{equation}
  G(0)  = -\frac {\beta^2}{4}
  \int_{\R^d} t(q)^2  g_0(\beta(q^2-\mu))\, \frac{dq}{(2\pi)^d}\,,
\end{equation}
which gives rise to the expression in (\ref{def:e1}).  Using
integration by parts,
\begin{align}\nn
  (p\cdot \nabla)^2G(0) & = - \frac {\beta^2}{8} \int_{\R^d} t(q)
  \left[ (p\cdot \nabla)^2 t\right]\!(q)\, g_0(\beta(q^2-\mu))\,
  \frac{dq}{(2\pi)^d} \\ \nn & \quad + \frac {\beta^4}4 \int_{\R^d}
  (p\cdot q)^2\, t(q)^2\, g_2(\beta(q^2-\mu))\, \frac{dq}{(2\pi)^d} \\
  & \quad + \frac {\beta^3 p^2}{8} \int_{\R^d} t(q)^2\, g_1(\beta(q^2
  -\mu))\, \frac{dq}{(2\pi)^d}\,. \label{849}
\end{align}

We proceed with $I_1^b$ and insert (\ref{difk}).  We treat
successively the terms proportional to $W$, to $A^2$, and to
$A$. After doing the contour integration and evaluating the trace per
unit volume, the six terms proportional to $W$ are
\begin{equation}
  h^{4-d} \sum_{p_1,p_2 \in (2\pi \Z)^d} \widehat \psi^*(p_1) \widehat \psi(p_2) \widehat W(-p_1-p_2) L(hp_1,hp_2) \,, 
\end{equation}
where
\begin{equation}
  L(p_1,p_2) = \beta^3 \int_{\R^d}  L(p_1,p_2,q) \, 
  \frac{dq}{(2\pi)^d} 
\end{equation}
and
\begin{align} L(p_1,p_2,q) & = \frac 1{4} \left( t(q) +
    t(q+p_1)\right)\left( t(q+p_1) + t(q+p_1+p_2)\right) \\ \nn &
  \quad \times \left( \left[a_0,a_0,a_2,-a_1\right]_f +
    \left[a_0,a_2,a_2,-a_1\right]_f + \left[a_0,a_2,-a_1,-a_1\right]_f
  \right) \,.
\end{align}
Since
\begin{equation} [a,a,a,-a]_f = \frac 1 8 g_1(a)
\end{equation}
we have
\begin{equation}
  L(0,0) = \frac  {\beta^3} 4 \int_{\R^d} t(q)^2 \, g_1(\beta(q^2-\mu)) \, \frac{dq}{(2\pi)^d} \,.
\end{equation}
In the same way as in Lemma~\ref{lem11}, one proves that

\begin{Lemma}
  The function $L$ is bounded, with
  \begin{equation}
    \left| L(p_1,p_2) - L(0,0) \right| \leq C \left( p_1^2+p_2^2\right) 
  \end{equation}
  for some constant $C>0$.
\end{Lemma}

The contribution of this term is thus
\begin{equation}\label{856}
  h^{4-d} L(0,0) \int_\calC W(x) |\psi(x)|^2\, dx  +  O(h^{6-d})  \|\psi\|_{H^2}^2 \,,
\end{equation}
where we have used a simple Schwarz inequality to bound the error
term, as well as the fact that $W \in L^2(\calC)$ and that $|\widehat
\psi(p)| \leq C \|\psi\|_{H^2}$ for all $p\in (2\pi\Z)^d$.

We obtain the same contribution with $A(x)^2$ in place of $W(x)$. The
terms linear in $A$ are given by
\begin{equation}
  h^{3-d} \sum_{p_1,p_2 \in (2\pi \Z)^d} \widehat \psi^*(p_1) \widehat \psi(p_2) \widehat A(-p_1-p_2)\cdot \int_{\R^d} L(hp_1,hp_2,q)(2q+hp_1+hp_2)\, \frac{dq}{(2\pi)^d}\,.
\end{equation}
For small $h$, we have
\begin{equation}
  \int_{\R^d} L(hp_1,hp_2,q)(2q+hp_1+hp_2)\, \frac{dq}{(2\pi)^d} = h\, \mathbb{K} (p_2-p_1) + O(h^2) \,,
\end{equation}
where $\mathbb{K}$ is a symmetric $d\times d$ matrix with entries
\begin{equation}
  \mathbb{K}_{ij} = \tfrac 12 L(0,0)\delta_{ij} + \frac {\beta ^4}4  \int_{\R^d} t(q)^2 q_i q_j \,g_2(\beta(q^2-\mu))\, \frac{dq}{(2\pi)^d} \,.
\end{equation}
Hence the leading term is
\begin{equation}
  h^{4-d} \, 2\, \re \int_{\calC} \psi(x)^* A(x)\cdot \mathbb{K} \left(- i \nabla\psi(x)\right) \, dx \,.
\end{equation}
The remainder can be bounded by
\begin{equation}
  C h^{5-d} \sum_{p_1,p_2 \in (2\pi \Z)^d} \left| \widehat \psi^*(p_1) \widehat \psi(p_2) \widehat A(-p_1-p_2)\right| \left(p_1^2 + p_2^2\right) \leq C h^{5-d} \|\psi\|_{H^1(\calC)}^2
\end{equation}
using that $|\widehat A(p)| |p|$ is summable by assumption.

We proceed with $I_1^c$.  First, there is the contribution of the
terms linear in $A$ in both factors $k-k_0$. These we have to
calculate explicitly to leading order in $h$, and then bound the
remainder. They contribute
\begin{equation}\label{862}
  h^{4-d} \frac{\beta^4}2 \sum_{i,j=1}^d  \int_{\calC} |\psi(x)|^2 A_i(x)A_j(x)\, dx\, \int_{\R^d} t(q)^2 q_i q_j \, g_2(\beta(q^2-\mu))\, \frac{dq}{(2\pi)^d}  
\end{equation}
plus terms of order $O(h^{5-d}) \|\psi\|_{H^1}^2$ and higher. In
combination, (\ref{856}) for $W$ and $A^2$, (\ref{862}) and
(\ref{849}) give all the terms in (\ref{def:e2}) except for the last,
which came from $F(0,0,0)$ in (\ref{f000}).

For the other terms, we use (\ref{resb1}) and (\ref{resb2}), and bound
all expressions involving $z$ with $\infty$-norms. We always get
enough decay in either the positive or the negative
$z$-direction. Only decay in one direction is needed, as $z$ can be
replaced by $-z$, which follows from the fact that the resulting
expressions are zero if $f(\beta z)$ is replaced by $f(\beta z) -
f(-\beta z) = \beta z$. (Compare with the discussion at the end of the
proof of Lemma~\ref{lem:i1a} in the appendix.) The same is true for
$I_1^d$. We omit the details. The final result is then that
\begin{equation}
  \frac 1 {2\pi i} \Tr \int_\Gamma f(\beta z) \left(I_1^c + I_1^d\right) dz = (\ref{862}) + O(h^{5-d}) \|\psi\|_{H^1}^2\,.
\end{equation}
This completes the proof of Theorem~\ref{thm:scl}.
 
\section{Proof of Theorem~\ref{lem3}}\label{sec:propproof}

As in the proof of Theorem~\ref{thm:scl}, we restrict out attention to
the case $d=3$, and leave the cases $d=1$ and $d=2$ to the reader.

Since the function
\begin{equation}
  \rho(z) = \frac 1 {1+ e^{z}}
\end{equation}
is analytic in the strip $|\im z| < \pi$, we can write $\alpha_\Delta$
with the aid of a contour integral representation as
\begin{equation}\label{rep:alp}
  \alpha_\Delta = \frac 1{2\pi i} \int_\Gamma  \rho(\beta z) \left[ \frac 1{z-H_\Delta}\right]_{12}dz\,,
\end{equation}
where $\Gamma$ is the contour $\im z = \pm \pi/(2\beta)$, and $[\,
\cdot\,]_{12}$ denotes the upper off-diagonal entry of an
operator-valued $2\times 2$ matrix. The integral in (\ref{rep:alp})
has to be suitably understood as a weak limit, similarly to
(\ref{int:rep}).

Using the resolvent identity and the definitions of $\Delta$ and
$\varphi$ in (\ref{def:delta}) and (\ref{def:varphi}), respectively,
we find that
\begin{equation}
  \alpha_\Delta  = \frac{h}2 \left( \psi(x) \varphi(-i h \nabla ) + \varphi(-ih \nabla)  \psi(x)\right)   + \sum_{j=1}^3 \eta_j \,,
\end{equation}
where
\begin{equation}\label{def:eta1}
  \eta_1 =  \frac h{4\pi i}\int_\Gamma \rho(\beta z) \left(\frac 1{z-k_0}\left[\psi, k_0\right] \frac{t}{z^2-k_0^2} + \frac{t}{z^2-k_0^2} [\psi,k_0]\frac 1{z+k_0} \right) \, dz \,,
\end{equation}
\begin{equation}
  \eta_2 =  \frac 1{2\pi i } \int_\Gamma \rho(\beta z) \frac 1{z-k_0} \left( (k-k_0) \frac 1{z-k}\Delta + \Delta\frac 1{z+k_0} (k_0 - \bar k)\right) \frac 1{z+\bar k} \, dz
\end{equation}
and
\begin{equation}
  \eta_3 = \frac 1{2\pi i}  \int_\Gamma \rho(\beta z) \frac 1{z-k}\Delta \frac 1{z+\bar k} \Delta^\dagger \frac 1{z-k} \Delta  \left[ \frac 1{ z-H_\Delta } \right]_{22} dz \,,
\end{equation}
with $t$ being short for the operator $t(-ih\nabla)$.  In the
following three lemmas we give bounds on these three operators.

\begin{Lemma}\label{lem13}
  \begin{equation}
    \|\eta_1\|_{H^1} \leq C h^{3-d/2} \|\psi\|_{H^2(\calC)}\,.
  \end{equation}
\end{Lemma}

\begin{proof}
  The $H^1$ norm of $\eta_1$ is given by
  \begin{equation}
    \|\eta_1\|_{H^1}^2 = h^{2-d} \sum_{p\in (2\pi \Z)^d} |\widehat \psi(p)|^2 J(hp) \,,
  \end{equation}
  where
  \begin{align}\nn
    J(p) = \frac {\beta^4} 4 \int_{\R^d} & \left( (q+p)^2 -
      q^2\right)^2 \left( 1+ (q+p)^2\right) \\ & \times \big| t(q)
    [a_0,-a_0,a_1]_\rho + t(p+q)[-a_0,a_1,-a_1]_\rho \big|^2 \,
    \frac{dq}{(2\pi)^d} \,. \label{def:j}
  \end{align}
  Here, we use again the notation $[\,\cdot\, ]_\rho$ for the divided
  differences (\ref{def:re})--(\ref{def:rec}), and $a_0 =
  \beta(q^2-\mu)$ and $a_1 = \beta((q+p)^2-\mu)$, as in the proof of
  Theorem~\ref{thm:scl}.  Since
  \begin{equation}\label{antis} [a,a,-a]_\rho = - [a,-a,-a]_\rho =
    \frac{ \rho(-a) - \rho(a)}{4 a^2} + \frac{ \rho'(a)}{2a} = \frac{1+2
      a e^ a - e^{2 a}}{4 a^2 (1+e^a)^2} \,,
  \end{equation}
  the integrand in (\ref{def:j}) vanishes like the fourth power of $p$
  for small $p$. In fact, using Lemma~\ref{lem:dec} and our
  assumptions on the regularity of $t$ it is easy to see that
  $J(p)\leq C |p|^4$, which yields the desired bound.
\end{proof}

We proceed with a bound on the $H^1$ norm of $\eta_2$.
\begin{Lemma}
  \begin{equation}
    \|\eta_2\|_{H^1} \leq C h^{3-d/2} \|\psi\|_{H^1(\calC)}\,.
  \end{equation}
\end{Lemma}

\begin{proof}
  We split $\eta_2$ into three parts, $\eta_2^a$, $\eta_2^b$ and
  $\eta_2^c$, where
  \begin{align}\nn
    \eta_2^a = - \frac {h^2}{2\pi} \int_\Gamma \rho(\beta z) \frac
    1{z-k_0} \Biggl( & (\nabla\cdot A + A\cdot \nabla) \frac
    1{z-k_0}\Delta \\ & + \Delta\frac 1{z+k_0} (\nabla\cdot A +A\cdot
    \nabla)\Biggl) \frac 1{z+ k_0} \, dz\,,
  \end{align}
  \begin{align}\nn
    \eta_2^b = \frac {h^2}{2\pi i } \int_\Gamma\rho(\beta z) \frac
    1{z-k_0}& \Biggl( (-i\nabla\cdot A - i h A\cdot \nabla) \frac
    1{z-k_0}(k-k_0)\frac 1{z-k}\Delta \\ \nn & + (-i\nabla\cdot A - i
    A \cdot\nabla) \frac 1{z-k_0}\Delta\frac 1{z+ k_0}(k_0-\bar k) \\
    & + \Delta\frac 1{z+k_0} (-i\nabla\cdot A - i A\cdot \nabla) \frac
    1{z+ k_0}(k_0-\bar k)\Biggl) \frac 1{z+\bar k} \, dz
  \end{align}
  and
  \begin{equation}
    \eta_2^c =  \frac {h^2}{2\pi i } \int_\Gamma \rho(\beta z) \frac 1{z-k_0} \left( (W+|A|^2) \frac 1{z-k}\Delta - \Delta\frac 1{z+k_0} (W+|A|^2)\right) \frac 1{z+\bar k} \,.
  \end{equation}

  We start by considering the first term in $\eta_2^c$. The square of
  the $H^1$ norm of the integrand can be bounded by
  \begin{multline}
    \Tr \left[ \frac{ 1-h^2\nabla^2}{|z-k_0|^2} (W+|A|^2)
      \frac{1}{z-k}\Delta \frac 1{|z+\bar k|^2} \Delta^\dagger \frac
      1{\bar z-k}(W+|A|^2) \right] \\ \leq \|\Delta\|_6^2
    \|W+|A|^2\|_\infty^2 \left\| \frac{
        1-h^2\nabla^2}{|z-k_0|^2}\right\|_\infty
    \left\|(z-k)^{-1}\right\|_\infty^2 \left\|(z+\bar
      k)^{-1}\right\|_3^2\,.
  \end{multline}
  The latter expression decays like $|z|^{-3+(d-2)/3}$ for negative
  $\re z$. Hence the resulting contour integral is finite, and we
  arrive at the bound $C h^{2-d/3} \|\Delta\|_6\leq C h^{3-d/2}
  \|\psi\|_{H^1(\calC)}$ for the $H^1$ norm.

  For the second term in $\eta_2^c$ we proceed similarly. It is
  important to first notice that
  \begin{equation}
    \frac 1{1+e^{\beta z}} = 1 - \frac 1{1+e^{-\beta z}}
  \end{equation}
  however, and that the $1$ does not contribute anything but
  integrates to zero. Proceeding as above, we conclude that
  \begin{equation}
    \|\eta_2^c\|_{H^1}^2 \leq C  h^{6-d} \|\psi\|_{H^1(\calC)}^2\,.
  \end{equation}
  The $H^1$ norm of $\eta_2^b$ can be bounded in essentially the same
  way, and we omit the details. The result is that also
  \begin{equation}
    \|\eta_2^b\|_{H^1}^2 \leq C  h^{6-d} \|\psi\|_{H^1(\calC)}^2
  \end{equation}
  holds

  At first sight, the $H^1$ norm of $\eta_2^a$ appears to be of order
  $h^{2-d/2}$ instead of $h^{3-d/2}$. Note, however, that if we
  commute $A(x)$ and $\psi(x)$ to the left in both terms, the
  resulting integral is zero due to the antisymmetry
  (\ref{antis}). Hence $\eta_2^a$ can be written as an integral over
  sums of various terms involving commutators of $A(x)$ and $\psi(x)$
  with functions of $-i\nabla$, and hence an additional factor of $h$
  is gained. In fact, the $H^1$ norm of $\eta_2^a$ can be written as
  \begin{equation}
    \|\eta_2^a\|_{H^1}^2 = h^{4-d}\sum_{p_1,p_2,p_3\in (2\pi\Z)^d} \widehat \psi(p_1) \widehat\psi^*(p_2)  \widehat A(p_3)\cdot \mathbb{M}(hp_1,hp_2,hp_3) \widehat A (-p_1-p_2-p_3)  \,,
  \end{equation}
  where $\mathbb{M}$ is a $d\times d$ matrix-valued function which can
  easily be calculated explicitly.  An analysis as in the proof of
  Lemma~\ref{lem13}, using the antisymmetry (\ref{antis}), shows that
  \begin{equation}
    \| \mathbb{M} (p_1,p_2,p_3) \| \leq C \left( p_1^2 +p_2^2 +p_3^2\right)\,.
  \end{equation}
  Under our assumption that $|\widehat A(p)|(1+|p|)$ is summable, this
  implies the result
  \begin{equation}
    \|\eta_2^a\|_{H^1} \leq C  h^{3-d/2} \|\psi\|_{H^1(\calC)}\,.
  \end{equation}
\end{proof}

Finally we bound the $H^1$ norm of $\eta_3$.
\begin{Lemma}
  \begin{equation}
    \|\eta_3\|_{H^1} \leq  C h^{3-d/2} \|\psi\|_{H^1(\calC)}^3 \,.
  \end{equation}
\end{Lemma}

\begin{proof}
  Since $H_\Delta$ is self-adjoint, $[(z-H_\Delta)^{-1}]_{22}$ is
  bounded by $2\beta /\pi$ for $z \in \Gamma$. Using H\"older's
  inequality, we can bound the square of the $H^1$ norm of the
  integrand by
  \begin{equation}
    \frac {4\beta^2}{\pi^2} \|\Delta\|_6^6 \left\| (z-k)^{-1}\right\|_\infty^2 \left\| (z+k)^{-1}\right\|_\infty^2  \left\| \sqrt{1-h^2\nabla^2}(z-k)^{-1}\right\|_\infty^2 \,.
  \end{equation}
  This expression decays like $|z|^{-3}$ for $\re z<0$, and like
  $|z|^{-1}$ for $\re z>0$. Hence
  \begin{equation}
    \|\eta_3\|_{H^1}^2 \leq C  \|\Delta\|_6^6\leq C h^{6-d} \|\psi\|_{H^1(\calC)}^6 \,.
  \end{equation}
\end{proof}

This completes the proof of Theorem~\ref{lem3}.

\appendix

\section{Proof of Lemma~\ref{lem:i1a}}

In this appendix, we shall give the proof that the operator
\begin{equation}
  B = \int_{\Gamma} f(\beta z) \frac 1{z-k} \Delta \frac 1{z+\bar k} \Delta^\dagger \frac 1{z-k}\, dz 
\end{equation}
is locally trace class. As in the proof of Theorem~\ref{thm:scl}, we
will focus our attention on the case $d=3$. Using the resolvent
identity, we split $B$ into three parts, $B=\sum_{i=1}^3 B_i$, with
\begin{equation}
  B_1 = \int_{\Gamma} f(\beta z) \frac 1{z-k_0} \Delta \frac 1{z+ \bar k} \Delta^\dagger \frac 1{z-k_0}\, dz  \,,
\end{equation}
\begin{equation}
  B_2 = \int_{\Gamma} f(\beta z) \frac 1{z-k_0}(k-k_0) \frac 1{z-k} \Delta \frac 1{z+\bar k} \Delta^\dagger \frac 1{z-k}\, dz \,,
\end{equation}
and
\begin{equation}
  B_3 = \int_{\Gamma} f(\beta z) \frac 1{z-k_0} \Delta \frac 1{z+\bar k} \Delta^\dagger \frac 1{z-k}(k-k_0) \frac 1{z-k_0}\, dz \,.
\end{equation}
We shall show that each individual piece is locally trace class.

We shall start with $B_2$. Using H\"older's inequality, the trace norm
of the integrand is bounded by $|f(\beta z)|$ times
\begin{equation}\label{a5}
  \left\| \frac 1{z-k_0}(k-k_0) \right\|_\infty \|\Delta\|_6^2 \left\| (z-k)^{-1} \right\|_p\left\| (z-k_0)^{-1} \right\|_p \left\| (z+\bar k)^{-1} \right\|_q
\end{equation}
with $2/p + 1/q = 2/3$. We shall choose $q^{-1} = 2/3-\epsilon$ for
some $0<\epsilon< 1/6$. The bounds (\ref{resb1})--(\ref{pb}) then
imply that (\ref{a5}) increases as $|z|^{1/2-\epsilon}$ for $\re z >
0$, and decreases like $|z|^{-13/6+\epsilon}$ for $\re z <0$. Since
$|f(\beta z)|$ decays exponentially for $\re z>0$ and increases
linearly for $\re z <0$, the contour integral is absolutely
convergent. This shows that $B_2$ is locally trace class.  In exactly
the same way, one shows that also $B_3$ is locally trace class.

For the remaining term $B_1$ it is not possible to take the trace norm
of the integrand and obtain a convergent integral. Hence we have to
argue differently. We again use the resolvent identity to write
\begin{align}\nonumber
  \frac 1{z + \bar k} & = \frac 1{z + k_0} + \frac 1{z+k_0} ( k_0 -
  \bar k) \frac 1{z+k_0} + \frac 1{z+k_0} ( k_0 - \bar k) \frac 1{z+
    k_0} ( k_0 - \bar k) \frac 1{z+k_0} \\ & \quad + \frac 1{z+k_0} (
  k_0 - \bar k) \frac 1{z+ k_0} ( k_0 - \bar k) \frac 1{z+k_0}
  (k_0-\bar k) \frac 1{z+\bar k}
\end{align}
and, correspondingly, write $B_1= B_1^a+B_1^b+B_1^c+B_1^d$.

We start with $B_1^a$.  Recall that $\Delta$ is the sum of two terms,
\begin{equation}
  \Delta  = -\frac h2( \psi(x) t(-ih\nabla) + t(-ih\nabla)\psi(x)) \,.
\end{equation}
By expanding $\psi$ in a Fourier series, we can use the triangle
inequality to bound
\begin{equation}
  \|B_1^a\|_1 \leq \frac {h^2}4 \sum_{p_1,p_2 \in( 2\pi \Z)^d} |\widehat \psi(p_1)| |\widehat \psi(p_2)|  \| D_{p_1,p_2} \|_1  
\end{equation}
where
\begin{align}\nonumber
  D_{p_1,p_2} = \int_{\Gamma} & f(\beta z) e^{-ip_1\cdot x} \frac
  1{z-k_0} \left( e^{ip_1\cdot x} t(-ih\nabla) + t(-ih\nabla) e^{ip_1
      \cdot x} \right) \frac{1}{z+k_0} \\ & \times \left( t(-ih\nabla)
    e^{-ip_2 \cdot x} + e^{-ip_2\cdot x}t(-ih\nabla)\right) \frac
  1{z-k_0} e^{ip_2\cdot x} \, dz \,.
\end{align}
Note that $D_{p_1,p_2}$ is translation invariant. That is, $D_{p_1,p_2}= D_{p_1,p_2}(-i h\nabla)$ is a multiplication operator in Fourier space, given by 
\begin{align}\nonumber
 D_{p_1,p_2}(q) =  & \left( t(q) + t(q+hp_1)
  \right) \left( t(q) + t(q+hp_2) \right) \\ & \!\! \times 
    \left[ -\beta (q^2-\mu), \beta ((q+hp_1)^2-\mu), \beta((q+hp_2)^2-\mu) \right]_f
   \,.
\end{align}
Its local trace norm thus
equals 
\begin{equation}
  \| D_{p_1,p_2} \|_1 = h^{-d} \int_{\R^d} \left| D_{p_1,p_2}(q) \right| \frac{dq}{(2\pi)^d} \,.
\end{equation}
A simple Schwarz inequality and an application of (\ref{eq:lem:dec1})
show that
\begin{equation}
  \| D_{p_1,p_2} \|_1 \leq C h^{-d} \int_{\R^d} \frac{ t(q)^2}{1+q^2} dq\,,
\end{equation}
independently of $p_1$ and $p_2$.  Finally, we can bound $\sum_p
|\widehat \psi(p)| \leq C \|\psi\|_{H^2}$ for $d\leq 3$.  We have thus
shown that
\begin{equation}
  \|B_1^a\|_1 \leq C h^{2-d} \|\psi\|_{H^2}^2  \int_{\R^d} \frac{t(q)^2}{1+q^2} dq \,. 
\end{equation}
This concludes the bound on the trace norm of $B_1^a$.

The term $B_1^b$ will be treated very similarly. Recall from
(\ref{difk}) that $k_0-\bar k$ is given by the sum of four terms. Let
us first look at the two terms linear in $A$. We expand both $\psi$
and $A$ in a Fourier series, and conclude that
\begin{equation}
  \|B_1^b\|_1 \leq \frac{h^3}{4} \sum_{p_1,p_2,p_3 \in( 2\pi \Z)^d} |\widehat \psi(p_1)| |\widehat \psi(p_2)| |\widehat A(p_3)|   \| G_{p_1,p_2,p_3} \|_1
\end{equation}
where $G_{p_1,p_2,p_3}(-ih\nabla)$ is a multiplication operator in
Fourier space, given by 
\begin{align}\nonumber
&  G_{p_1,p_2,p_3}(q)=  \\ \nn &  \left( t(q) + t(q+hp_1) \right)\left( t(q - h
    p_3 ) + t(q+h p_2- hp_3) \right)\left( 2 q - hp_3 \right) \\ &
   \times \!\! \left[\! -\beta (q^2-\mu),- \beta((q -hp_3)^2-\mu) , \beta((q +
    hp_1)^2-\mu), \beta((q + h(p_2 -  p_3))^2-\mu) \!\right]_f .
\end{align}
To bound its trace norm, we can again use the bounds on the divided
differences given in Lemma~\ref{lem:dec}. We shall also need the
following additional bound.  From the definition (\ref{def:rec}) it
follows that if $f$ is Lipschitz continuous, then
\begin{equation}
  \left| \left[a_1,a_2,a_3,a_4\right]_f \right| \leq \frac{C}{|a_1-a_4|}\left( \frac{1}{|a_1-a_3|} + \frac 1{|a_2-a_4|}\right)\,.
\end{equation}
If we apply this with $a_1 \leq a_2 \leq 0 \leq a_3 \leq a_4$, we
obtain the upper bound $C/(b_1 b_2)$, where $b_1$ and $b_2$ are,
respectively, the largest and second largest among the numbers
$\{|a_1|,|a_2|,|a_3|,|a_4|\}$.  In combination with
(\ref{eq:lem:dec1}) one easily obtains the bound 
\begin{align}\nonumber
  & \left| \! \left[\! -\beta (q^2-\mu),- \beta((q -hp_3)^2-\mu) , \beta((q + hp_1)^2-\mu),
      \beta((q + h(p_2- p_3))^2-\mu) \!\right]_f \!\right| \\ & \leq
  \frac{C}{1+q^2} \frac 1{1 + \max\left\{ (q -hp_3)^2 , (q + hp_1)^2,
      (q + hp_2 - h p_3)^2 \right\}}\,.
\end{align}
A simple Schwarz inequality then shows that the trace norm of
$G_{p_1,p_2,p_3}$ is bounded by $C (1+|p_3|)$ for some constant $C$
independent of $p_1$, $p_2$ and $p_3$. Since both $|\widehat \psi(p)|$
and $|\widehat A(p)|(1+|p|)$ are summable, this shows that this part
of $B_1^b$ is locally trace class. The other part, where the term
quadratic in $A$ and the term with $W$ in $k_0-\bar k$ are taken into
account, can be treated in the same way, using our assumption that
$|\widehat W(p)|$ is summable.

The same method can be used to show that $B_1^c$ is locally trace
class. One simply expands all multiplication operators in Fourier
series, and then bounds the resulting translation invariant
operator. We omit the details.

We are left with the term $B_1^d$. We first note that since
\begin{equation}
  f(\beta z) = f(-\beta z) + \beta z
\end{equation}
we can replace $f(\beta z)$ by $f(-\beta z)$ in the integrand without
changing the integral, since the expression with $\beta z$ in place of
$f(\beta z)$ integrates to zero, as can be easily be seen from
(\ref{83}), for instance.  We can bound the trace norm of the
integrand with the aid of H\"older's inequality and
(\ref{resb1})--(\ref{pb}) by a constant times $|z|^{-13/6}$ for $\re
z>0$, and $|z|^{1/2}$ for $\re z<0$, respectively. Since we are
integrating against a function that decays exponentially for negative
$\re z$ and increases linearly for positive $\re z$, the integral is
absolutely convergent. This proves that $B_1^d$ is locally trace
class.

\bigskip

\noindent {\it Acknowledgments.} Part of this work was carried out at
the Erwin Schr\"odinger Institute for Mathematical Physics in Vienna,
Austria, and the authors are grateful for the support and hospitality
during their visit. R.S. would also like to thank the Departamento de
Fisica at the Pontificia Universidad Cat{\'o}lica de Chile for their
hospitality. Financial support via U.S. NSF grants DMS-0800906 (C.H.)
and PHY-0845292 (R.S.) and a grant from the Danish council for
independent research (J.P.S.) is gratefully acknowledged.



\begin{thebibliography}{20}

\bibitem{AS} M. Abramowitz, I.A. Stegun, {\it Handbook of mathematical
    functions}, Dover (1964).

\bibitem{BCS} J. Bardeen, L. Cooper, J. Schrieffer, {\it Theory of
    Superconductivity}, Phys. Rev. {\bf 108}, 1175--1204 (1957).

\bibitem{dono} W.F. Donoghue, {\it Monotone Matrix Functions and
    Analytic Continuation}, Springer (1974).

\bibitem{FH} S. Fournais, B. Helffer, {\it Spectral Methods in Surface
    Superconductivity}, Birkh\"auser (2010).

\bibitem{FHNS} R.L. Frank, C. Hainzl, S. Naboko, R. Seiringer, {\it
    The critical temperature for the BCS equation at weak coupling},
  J. Geom. Anal.  {\bf 17}, 559--568 (2007).

\bibitem{deGenne} P.G. de Gennes, {\it Superconductivity of Metals and
    Alloys}, Westview Press (1966).

\bibitem{GL} V.L. Ginzburg, L.D. Landau, {\it On the theory of
    superconductivity}, Zh. Eksp. Teor. Fiz. {\bf 20}, 1064--1082
  (1950).

\bibitem{gorkov} L.P. Gor'kov, {\it Microscopic derivation of the
    Ginzburg-Landau equations in the theory of superconductivity},
  Zh. Eksp. Teor. Fiz. {\bf 36}, 1918--1923 (1959); {\it English
    translation} Soviet Phys. JETP {\bf 9}, 1364--1367 (1959).

\bibitem{HHSS} C. Hainzl, E. Hamza, R. Seiringer, J.P. Solovej, {\it
    The BCS functional for general pair interactions},
  Commun. Math. Phys. {\bf 281}, 349--367 (2008).

\bibitem{HLS} C. Hainzl, M. Lewin, R. Seiringer, {\em A nonlinear
    theory for relativistic electrons at positive temperature},
  Rev. Math. Phys. {\bf 20}, 1283--1307 (2008).

\bibitem{HS} C. Hainzl, R. Seiringer, {\it Critical temperature and
    energy gap in the BCS equation}, Phys. Rev. B {\bf 77}, 184517
  (2008).

\bibitem{HS3} C. Hainzl, R. Seiringer, {\em The BCS critical
    temperature for potentials with negative scattering length.}
  Lett. Math. Phys. {\bf 84}, 99--107 (2008).

\bibitem{HS2} C. Hainzl, R. Seiringer, {\it Spectral properties of the
    BCS gap equation of superfluidity}, in: Mathematical Results in
  Quantum Mechanics, proceedings of QMath10, Moeciu, Romania,
  September 10--15, 2007, World Scientific (2008).

\bibitem{helffer} B. Helffer, D. Robert, {\it Calcul fonctionnel par la 
transformation de Mellin et op{\'e}rateurs admissibles}, J. Funct. Anal. {\bf 53}, 
246--268 (1983).

\bibitem{Leg} A.J. Leggett, {\it Diatomic Molecules and Cooper Pairs},
  in: Modern trends in the theory of condensed matter, A. Pekalski,
  R. Przystawa, eds., Springer (1980).

\bibitem{bookleggett} A.J. Leggett, {\it Quantum Liquids}, Oxford
  (2006).

\bibitem{LL} E.H. Lieb, M. Loss, {\it Analysis}, $2^{\rm nd}$ ed.,
  Amer. Math. Soc. (2001).

\bibitem{LSstab} E.H. Lieb, R. Seiringer, {\it The Stability of Matter
    in Quantum Mechanics}, Cambridge (2010).

\bibitem{ReSi} M. Reed, B. Simon, {\it Methods of Modern Mathematical
    Physics. IV. Analysis of Operators}, Academic Press (1978).

\bibitem{robert} D. Robert, {\it Autour de l'approximation semi-classique}, 
Progress in Mathematics {\bf  68}, Birkh\"auser (1987).

\bibitem{Serfaty} E. Sandier, S. Serfaty, {\it Vortices in the
    Magnetic Ginzburg-Landau Model}, Birkh\"auser (2006).

\bibitem{simon} B. Simon, {\it Trace ideals and their applications},
  $2^{\rm nd}$ ed., Amer. Math. Soc. (2005).

\bibitem{thirring} W. Thirring, {\it Quantum Mathematical Physics},
  $2^{\rm nd}$ ed., Springer (2002).

\end{thebibliography}
\end{document}